\documentclass[10pt,journal,final,twocolumn]{IEEEtran}
\DeclareFixedFont{\auacc}{OT1}{phv}{b}{it}{18}   
\DeclareFixedFont{\newauacc}{OT1}{ptm}{b}{rm}{12}   
\usepackage{hyperref}
\usepackage{cite}
\usepackage{graphicx}
\usepackage{subfigure}
\usepackage{url}
\usepackage{theorem}
\usepackage{amsmath}
\usepackage{amsfonts}
\usepackage{amssymb}
\usepackage{color}
\usepackage{paralist}
\usepackage{ifthen}

\newboolean{longver}
\setboolean{longver}{false}

\DeclareMathSymbol{\R}{\mathord}{AMSb}{"52}
\DeclareMathSymbol{\C}{\mathord}{AMSb}{"43}
\DeclareMathSymbol{\Z}{\mathord}{AMSb}{"5A}
\DeclareMathSymbol{\N}{\mathord}{AMSb}{"4E}
\DeclareMathSymbol{\K}{\mathord}{AMSb}{"4B}
\DeclareMathSymbol{\M}{\mathord}{AMSb}{"4D}
\DeclareMathSymbol{\Q}{\mathord}{AMSb}{"51}
\DeclareMathSymbol{\Lset}{\mathord}{AMSb}{"4C}

\theorembodyfont{\rmfamily}
\theoremheaderfont{\bfseries\rmfamily}
\newtheorem{theorem}{Theorem}

\theoremheaderfont{\rmfamily}

\newtheorem{corollary}{Corollary}
\newtheorem{lem}{Lemma}
\newtheorem{obs}{Observation}

\theorembodyfont{\upshape}

\newtheorem{remark}{{Remark}}

\newcommand\defeq{\mathrel{:=}}
\newcommand\disteq{\stackrel{\mathrm{d}}{=}}
\newcommand\distapprox{\stackrel{\mathrm{d}}{\approx}}
\newcommand\ceiltheta{{\lceil \alpha \rceil}}
\newcommand\ie{{\em i.e.}}

\newcommand\etal{{\em et al.}}
\definecolor{webmag}{rgb}{0.5,0,0.5}
\newcommand{\nop}[1]{}

\title{Fundamentals of the Backoff Process in 802.11: \\ Dichotomy of the Aggregation}
\author{\authorblockN{Jeong-woo~Cho~and~Yuming~Jiang}
\thanks{This work was supported in part by ``Centre for Quantifiable Quality of Service in Communication Systems, Centre of Excellence'' appointed by The Research Council of Norway, and funded by The Research Council, NTNU and UNINETT. A part of this work was done when J. Cho was with EPFL, Switzerland. A preliminary abstract version of this work appeared at ACM SIGMETRICS Workshop on Mathematical Performance Modeling and Analysis (MAMA'09).}
\thanks{J. Cho and Y. Jiang are with the Centre for Quantifiable Quality of Service in Communication Systems, Norwegian University of Science and Technology (NTNU), NO-7491 Trondheim, Norway (email: \{jeongwoo,jiang\}@q2s.ntnu.no).}
}

\makeatletter
\let\@copyrightspace\relax
\makeatother

\begin{document}
\maketitle

\newcommand{\expectation}{\textsf{E}}
\newcommand{\probability}{\textsf{P}}
\newcommand{\pdf}{\textsf{f}}
\newcommand{\slow}{\ell}
\newcommand{\nextline}{\mbox{}\\}
\newcommand{\ud}{\mathrm{d}}
\newcounter{tempcounter}
\newcounter{acounter}

\begin{abstract}
This paper discovers fundamental principles of the backoff process that governs the performance of IEEE 802.11. A simplistic principle founded upon regular variation theory is that the backoff time has a truncated {\itshape Pareto-type} tail distribution with an exponent of $\boldsymbol{ (\log \gamma)/\log m}$ ($\boldsymbol{m}$ is the multiplicative factor and $\boldsymbol{\gamma}$ is the collision probability). This reveals that the per-node backoff process is heavy-tailed in the strict sense for $\boldsymbol{ \gamma > 1/m^2}$, and paves the way for the following unifying result.

The state-of-the-art theory on the superposition of the heavy-tailed processes is applied to establish a {\itshape dichotomy} exhibited by the aggregate backoff process, putting emphasis on the importance of time-scales on which we view the backoff processes. While the aggregation on normal time-scales leads to a Poisson process, it is approximated by a new limiting process possessing {\itshape long-range dependence} (LRD) on coarse time-scales. This dichotomy turns out to be instrumental in formulating short-term fairness, extending existing formulas to arbitrary population, and to elucidate the absence of LRD in practical situations. A refined wavelet analysis is conducted to strengthen this argument.

\end{abstract}


\section{Introduction}\label{sec:intro}
Since its introduction, the performance of IEEE 802.11 has attracted a lot of research attention and the center of the attention has been the {\it throughput} \cite{refBianchi, refKumar07}. Recently, other critical performance aspects of 802.11 also burst onto the scene, which include {\it short-term fairness} \cite{refKoksal,refMarkusExp} and {\it delay} \cite{refTickoo}. It goes without saying that there has been a phenomenal growth of Skype and IPTV users \cite{refSkype,refIPTV} and it is reported in \cite{refPew} that an ever-increasing percentage of these users connects to the Internet through wireless connections in US. Remarkably, it is found in \cite{refSkype} that {\it jitter} is more negatively correlated with Skype call duration than delay, \ie, Skype users tend to hang up their calls earlier with large jitters. This finding empirically testifies large jitter of access networks {\it annoys} Skype users, let alone QoS (quality of service). This quantified dissatisfaction of users provides a motivation for a thorough understanding of delay and jitter performance in 802.11.

For throughput analysis, Kumar \etal, in the seminal paper \cite{refKumar07}, axiomized several remarkable observations based on a {\it fixed point equation} (FPE), advancing the state of the art to more systematic models and paving the way for more comprehensive understanding of 802.11. Above all, one of the key findings of \cite{refKumar07}, already adopted in the field \cite{refKwak,refSakuraiDelay}, is that the full interference model\footnotemark, also called the single-cell model \cite{refKumar07}, in 802.11 networks leads to the {\it backoff synchrony property} \cite{refProutierePushing} which implies the backoff process can be completely separated and analyzed through the FPE technique. Another observation in \cite{refKumar07} was that if the collision probability $\gamma$ is constant, one can derive the so-called Bianchi's formula by appealing to renewal reward theorem \cite{refGibbs}, without the Markov chain analysis in \cite{refBianchi}.

\footnotetext{In the full interference or single-cell model, every node interferes with the rest of the nodes, \ie, its corresponding interference graph is fully connected.}

An intriguing notion, called {\it short-term fairness}, has been introduced in some recent works \cite{refKoksal,refBergerUniform,refMarkusExp}, defining $ \probability [z \vert \zeta]$ as the probability that other nodes transmit $z$ packets while a tagged node is transmitting $\zeta$ packets. It can be easily seen that this notion pertains to a purely backoff-related argument also owing to the backoff synchrony property in the full interference model \cite{refKumar07}. The two papers \cite{refBergerUniform,refMarkusExp}, in the course of deriving equations for $ \probability [z \vert \zeta]$, assumed that the summation of the backoff values generated per packet, which we denote by $\Omega$, is uniformly and exponentially distributed, respectively. Specifically, despite the same situation where {\it two} nodes contend for the medium, the former \cite{refBergerUniform} assumed that $\Omega$ is {\it uniformly} distributed because the initial backoff is uniformly distributed over the set $\left\{ 0, 1,\cdots,2b_0 -1 \right\}$ where $2b_0$ is the initial contention window and observed in \cite[Fig. 2]{refBergerUniform} that this assumption leads to a good match between the expression $ \probability [z \vert \zeta]$ derived under the uniform assumption on $\Omega$ and the testbed data measured in their experiments, while the latter \cite{refMarkusExp} also observed in \cite[Fig. 5(a)]{refMarkusExp} that the testbed data measured in their experiments closely match the expression $\probability [Z\vert \zeta]$ derived under the the exponential assumption on $\Omega$:
\vskip 0pt \noindent { \bfseries\textcolor{webmag}{Q1}:\mdseries\itshape ``What makes two different observations?''} (to be answered in Section \ref{sec:meanfieldfairness})
\vskip 0pt

In addition, the two works \cite{refBergerUniform,refMarkusExp} acquired the expression of $ \probability [z \vert \zeta]$ only for the two node case. A more general formula for arbitrary number of nodes should deepen our appreciation of short-term fairness. It is natural to ask the following pertinent questions:
\vskip 0pt \noindent { \bfseries\textcolor{webmag}{Q2}:\mdseries\itshape ``Can we develop a general model for short-term fairness?''} (to be answered in Corollaries \ref{th:aintertxprob} \& \ref{th:intertxprob})
\vskip 0pt

In proportion as people take a growing interest in the delay performance of 802.11, the number of fundamental questions that we face increases. In \cite{refAbdrabou}, it was argued based on simulation results that the access delay in 802.11 closely follows a Poisson distribution. They have shown that the number of successful packet transmissions by any node in the network over a time interval has a probability distribution that is close to Poisson by an upper bounded distribution distance. This raises an intriguing question:

\vskip 0pt \noindent { \bfseries\textcolor{webmag}{Q3}:\mdseries\itshape ``Is there a Poissonian property? If yes, what is the cause?''} (to be answered in Theorem \ref{th:poi})
\vskip 0pt

Another case in point is found in a recent work \cite{refSakuraiDelay} that extends the access delay analysis in the seminal paper of Kwak \etal \cite{refKwak} and makes an attempt at analyzing higher order moments by applying the FPE technique. One interesting finding in \cite{refSakuraiDelay} is that the access delay has a {\it wide-sense heavy-tailed} distribution \cite[Theorem 1]{refSakuraiDelay} which means that its moment generating function $\int_{0}^{\infty} {\rm e}^{tx} f(x) \ud x $ is $\infty$, $\forall t>0$, where $f(x)$ is the corresponding pdf (probability density function) \cite{refHeavyRolski}. One should be careful in interpreting this finding because the wide-sense heavy-tailedness does {\it not} imply strict sense heavy-tailedness, which roughly means the ccdf (complementary cumulative distribution function) is of Pareto-type \cite{refCrovella} with an exponent over $(-2,0)$. In fact, there are lots of distributions, namely, lognormal, Pareto, Cauchy and Weibull distributions, which belong to the class of wide-sense heavy-tailed distributions. Consequently, the discussion poses the following {\it challenge} which is undoubtedly a tantalizing question.

\vskip 0pt \noindent { \bfseries\textcolor{webmag}{Q4}:\mdseries\itshape ``What is the distribution type of the delay-related variables?''} (to be answered in Theorems \ref{th:heavy} \& \ref{th:powertail})
\vskip 0pt


Finally, it is, perhaps, surprising that long-range dependence of 802.11 has not been rigorously analyzed even for the single node case, not to mention the aggregate process of many nodes. One minor contribution of this paper is that we prove in Theorem \ref{th:powertail} that the individual arrival process (consisting of successful transmissions of one node) can be viewed as a renewal process with heavy-tailed inter-arrival times, which implies that the individual arrival process possesses long-range dependence simply by appealing to \cite{refLipsky}.

However, for the superposition arrival process (consisting of successful transmissions of all nodes), there is no clear answer. For example, Tickoo and Sikdar \cite{refTickooSS} conjectured the absence of long-range dependence of aggregate total load, which we call superposition arrival process. It is remarkable that the absence of long-range dependence has been also supported through empirical analysis such as wavelet-based method \cite{refWaveletLens} by Veres and Boda \cite{refChaoticVeres} in the context of TCP flows in wired networks. Since there is an analogy between the backoff mechanisms adopted by 802.11 and TCP (in wired networks) in that
\begin{compactenum}
\item both of them adopt backoff schemes (802.11) or retransmission scheme (TCP) where the probability of these events is either the collision probability (802.11) or the packet drop probability in router buffers (TCP),
    \item the mean of the backoff (contention window in 802.11) or retransmission time (timeout in TCP) doubles for each backoff or retransmission,
\end{compactenum}
one might wonder if there is a fundamental reason that elucidates these observations.

\vskip 0pt \noindent { \bfseries\textcolor{webmag}{Q5}:\mdseries\itshape ``Does the aggregate transmission process possess long-range dependence? If yes, why is it seldom observed?''} (to be answered in Theorem \ref{th:itp} and Section \ref{sec:longrange})
\vskip 0pt

The focus of this paper is on the backoff process in 802.11, since it plays the central role in quantifying the performance of 802.11 \cite{refKumar07}. For example, to grasp the heart of the delay properties, the backoff value distribution in 802.11 DCF (distributed coordination function) mode can be used as a {\it surrogate} for the access delay \cite{refKwak}. 
As discussed above, the throughput performance and short-term fairness performance also depend on the backoff process and are particularly affected by the backoff synchrony property. Essentially, once the backoff distribution is obtained, various performance aspects can be analyzed.

\subsection{Contributions of this work}
This paper discovers fundamental principles of the backoff process and provides answers to the open questions highlighted above, which constitute the contributions of the paper. Particularly, it turns out that we find out the answers to most aforementioned questions {\bfseries\textcolor{webmag}{Q2}}-{\bfseries\textcolor{webmag}{Q5}} in the course of deriving the following two principles based on a new methodology, \ie, point process approach.

\begin{compactitem}
\item {\bf Power-tail principle}: The per-packet backoff time distribution has a slowly-varying power-tail (Theorem \ref{th:powertail}).
\item {\bf Dichotomy of aggregation}: Depending on the time-scales on which the backoff processes are aggregated, the resultant process becomes either Poissonian or a new process (Theorems \ref{th:poi} \& \ref{th:itp}).
\end{compactitem}

The power-tail principle, which is derivable only after we accumulate a store of knowledge (Section \ref{sec:meanfieldfairness}, Lemma \ref{lem:meanexist} and Theorem \ref{th:heavy}), characterizes the backoff distribution in a tractable and simplistic way, owing to regular variation theory, answering {\bfseries\textcolor{webmag}{Q4}}.
The dichotomy of aggregation implies that, when we view the aggregate process on normal time-scales, owing to the tendency of each component process to become sparse as population grows, we observe only a Poissonian as its marginal distribution. However, viewed on coarse time-scales, the aggregate process is identified as a long-range dependent process. This rigid dichotomy is instrumental in finding answers to {\bfseries\textcolor{webmag}{Q2}}, {\bfseries\textcolor{webmag}{Q3}} and {\bfseries\textcolor{webmag}{Q5}}, and expatiates upon the coexistence of contrary properties suggested by {\bfseries\textcolor{webmag}{Q3}} and {\bfseries\textcolor{webmag}{Q5}}.
All the theorems in the paper are {\it closely linked} with each other, forming a solid framework for the performance analysis of 802.11. These results help us to get the complex details of the backoff process in 802.11 into perspective under one framework.

The rest of the paper is organized as follows. In Section \ref{sec:justification}, we revisit the Bianchi's formula along with a survey of recent advances in mean field theory, with which the analysis of the backoff process at one node can be decoupled from other nodes. 
In Section \ref{sec:meanfieldfairness}, we present the exact distribution of per-packet backoff. We establish in Section \ref{sec:intertxdist} that the aggregate backoff process can be approximated by a Poisson process under the large population regime. In Section \ref{sec:asymptotic}, we extend to asymptotic analysis and prove the power-tail principle. In Section \ref{sec:shortterm}, we first propose a new process approximation on a coarse time scale, which is then applied to formulate short-term fairness and to identify long-range dependence. After conducting a wavelet analysis on long-range dependence in Section \ref{sec:longrange}, we conclude this paper.

\section{Bianchi's Formula Revisited}\label{sec:justification}

The backoff process in 802.11 is governed by a few rules if the duration of per-stage backoff is taken to be exponential: (i) every node in backoff stage $k$ attempts transmission with probability $p_k$ for every time-slot; (ii) if it succeeds, $k$ changes to $0$; (iii) otherwise, $k$ changes to $(k+1)$ \verb#mod# $(K+1)$ where $K$ is the index of the highest backoff stage. Markov chain models, which have been widely used in describing complex systems including 802.11, however, very often lead to excessive complications as discussed in Section \ref{sec:intro}.
 In this section, we present a surrogate tool for the analysis, {\it mean field theory}. It is noteworthy that the rules used in 802.11, \ie, (i)--(iii), closely resemble the mean field equations laid out below.

\subsection{Basic Operation of DCF Mode}\label{sec:dcf}

 Time is slotted. Each node following the randomized access procedure of 802.11 distributed coordination function (DCF) generates a {\it backoff value} after receiving the Short Inter-Frame Space (SIFS) if it has a packet to send. This backoff value is {\it uniformly} distributed over $\{$0, 1, $\cdots$, 2${b_0}-$1$\}$ (or $\{$1, 2, $\cdots$, 2${b_0} \}$) where $2 b_0$ is the initial contention window.

 Whenever the medium is idle for the duration of a Distributed Inter-Frame Space (DIFS), a node unfreezes (starts) its countdown procedure of the backoff and decrements the backoff by one per every time-slot. It freezes the countdown procedure as soon as the medium becomes busy. There exist $K+1$ backoff stages whose indices belong to the set $\{0,1,\cdots,K\}$ where we assume $K>0$. If two or more wireless nodes finish their countdowns at the same time-slot, there occurs a collision between RTS (ready to send) packets if the CSMA/CA (carrier sense multiple access with collision avoidance) is implemented, otherwise two data packets collide with each other. If there is a collision, each node who participated in the collision multiplies its contention window by the multiplicative factor $m$. In other words, each node changes its backoff stage index $k$ to $k+1$ and adopts a new contention window $2b_{k+1} = 2 m^{k+1} \cdot b_0 $. If $k+1$ is greater than the index of the highest backoff stage number, $K$, the node steps back into the initial backoff stage whose contention window is set to $2 b_0$. In the IEEE 802.11b standard, $m=2$, $K=6$ ($7$ attempts per packet), and $2 b_0 = 32$ are used.

\ifthenelse{\boolean{longver}}{This work focuses on the performance of {\it single-cell} 802.11 networks in which all 802.11-compliant nodes are within such a distance of each other that a node can hear whatever the other nodes transmit. Since all nodes simultaneously freeze their backoff countdown during channel activity, the total time spent in backoff stages up to any time is the same for all nodes. Therefore, it is {\it sufficient} to analyze the backoff process in order to investigate the performance of single-cell networks.}{This work focuses on the performance of {\it single-cell} 802.11 networks where it is {\it sufficient} to analyze the backoff process in order to investigate the performance of single-cell networks}
\ifthenelse{\boolean{longver}}{This technique has been adopted in many works including \cite{refKumar07,refBordenaveMF,refJYMF}.}{}

\subsection{The Bianchi's Formula}\label{sec:revisited}

In performance analysis of 802.11, Bianchi's formula and its many variants are probably the most known \cite{refBianchi,refKumar07,refKwak,refSakuraiDelay}. Assuming that there are $N$ nodes, the Bianchi's formula can be written compactly in a more general {\it fixed point equation} (FPE) form:
\setcounter{tempcounter}{\value{equation}}
\setcounter{equation}{0}
\renewcommand{\theequation}{FPE\arabic{equation}}
\begin{align}\label{eq:fixedpoint}
\bar p  & =  \frac{\sum_{k=0}^{K} \gamma^k}{\sum_{k=0}^{K} \frac{\gamma^k}{q_k}}, \\
\gamma  & = 1 - {\rm e}^{- (N-1){\bar p}} \label{eq:fixedpointg}
\end{align}
where ${\bar p} $ and $\gamma$ respectively designate the average attempt rate and collision probability of every node at each time-slot. The attempt probability in backoff stage $k$ is denoted by $q_k$ and defined as the inverse of the mean contention window, \ie, $q_k = 2/(2b_k - 1)$. It satisfies $0<q_k \leq 1$ as $b_k \geq 1$. Note that Bianchi's formula holds under the well-known assumption:
\begin{compactenum}[\bfseries{A}.\arabic{acounter}]\setcounter{acounter}{1}
\item\addtocounter{acounter}{1} All the transmission queues of nodes are saturated.
\end{compactenum}

\renewcommand{\theequation}{\arabic{equation}}
\setcounter{equation}{\value{tempcounter}}

\ifthenelse{\boolean{longver}}{It is, perhaps, surprising that whether the formula \eqref{eq:fixedpoint} is valid or not has never been completely agreed upon despite the fact that the formula has been a {\it de facto} principal tool for the analysis.}{} Exactly under which condition \eqref{eq:fixedpoint} holds is recently being rediscovered with rigorous mathematical arguments \cite{refBordenaveMF,refSharmaScaledMarkov,refJYMF}, which, sometimes called {\it mean field approximation}. This fundamental approach was originally developed by Bordenave \etal \cite{refBordenaveMF} and Sharma \etal \cite{refSharmaScaledMarkov}. Remarkably, Bordenave \etal \cite{refBordenaveMF} adopted a generalized particle interaction model which encompasses Markovian evolution of the system other than particles at the same time. Bena\"im and Le Boudec \cite{refJYMF} overcame some limitations of the model \cite{refBordenaveMF}, broadening its applicability. The main result here is that, as the number of particles goes to infinity, \ie, $N\to \infty$, the state distribution of every node evolves according to a set of $K+1$ dimensional nonlinear ordinary {\it differential} equations under an appropriate scaling of time. Bena\"im and Le Boudec \cite{refJYMF} also observed that decoupling approximation represented by \eqref{eq:fixedpoint} does {\it not} hold if the differential equations does not have a unique globally attractor.

 Remarkably, Bordenave \etal \cite{refBordenaveMF} have proven that the differential equations are globally stable if $K=\infty$ and the re-scaled attempt probability $ Q_{k} \defeq N q_{k}$ satisfies $Q_{k+1} = Q_{k}/2$ with $ Q_0 < \ln 2$. In the meantime, Sharma \etal \cite{refSharmaScaledMarkov} obtained a result for $K=1$ and mentioned the difficulty to go beyond. However, the case for other finite $K$ has remained to be proved \cite[pp.833]{refJYMF}. Recently, we solved this issue to a large extent \cite{refChoValidity}  by proving that, for finite $K$, a simplistic condition $Q_k \leq 1$ (or $q_{k} \leq 1/N$) for all $k \in \{0, \cdots, K \}$ guarantees the global stability of the differential equations as well as the uniqueness of the solutions to the Bianchi's formula.  
 As we discussed in \cite{refChoValidity}, there are still many outstanding problems upon the stability of the associated differential equations.

While one of the aims of these efforts \cite{refBordenaveMF,refSharmaScaledMarkov,refJYMF,refChoValidity} is to identify the fundamental conditions under which the collision probability is deterministic and time-invariant for large population ($N=\infty$), once we assume the collision probability is such for $N<\infty$, the demonstration of the formula \eqref{eq:fixedpoint} is shown to be straightforward \cite{refKumar07}. That is to say, we need to make the following simple assumption.
\begin{compactenum}[\bfseries{A}.\arabic{acounter}]
\item\addtocounter{acounter}{1} For each node, conditional upon its transmission attempt, the collision events form an i.i.d. sequence, which is independent from other nodes.
\end{compactenum}
The observation in \cite{refKumar07} was that, under the above assumptions, one can easily derive \eqref{eq:fixedpoint} by appealing to renewal reward theorem \cite{refGibbs}, without the Markov chain analysis in \cite{refBianchi}. Thus from now on, the attempt probability is given by \eqref{eq:fixedpoint}. As a by-product, we can also see that the distribution of backoff stages, which we denote by $\phi_k$, $k \in \{ 0, \cdots, K \}$, takes the following form
\begin{align}
\phi_k  =  \frac{\gamma^k}{q_k}  \cdot \frac{1}{ \sum_{k=0}^{K} \frac{\gamma^k}{q_k}  }  . \label{eq:phiknew}
\end{align}

\ifthenelse{\boolean{longver}}{The expression of the collision probability \eqref{eq:fixedpointg} was first used in \cite[Section IV]{refKumar07} where it was shown that if the attempt probability of every node $\bar p$ is inversely proportional to $N$, \eqref{eq:fixedpointg} is implied by appealing to Poisson convergence theorem \cite{refPoisson}.}{The expression of the collision probability \eqref{eq:fixedpointg} was first used in \cite[Section IV]{refKumar07}.} A similar expression was also used in \cite{refBordenaveMF,refJYMF} under the intensity scaling, which means that the attempt probability of every node in any backoff stage is of the order of $1/N$. We use \eqref{eq:fixedpointg} instead of its original version in \cite{refBianchi} because, as argued in \cite{refBordenaveMF,refJYMF}, the approximation provided by \eqref{eq:fixedpointg} is well founded on a mean field result. Lastly, it is also noteworthy that the analysis in Section \ref{sec:meanfieldfairness} and \ref{sec:intertxdist} does not depend on whether $K$ is finite or not.

\section{Backoff Analysis}\label{sec:meanfieldfairness}

The backoff value distribution and the backoff stage distribution should not be confused in meaning. While the latter is the distribution of the backoff stage of a node, the former is the distribution of the backoff value generated for initiating the backoff countdown when the node has a packet to transmit. The backoff value distribution at backoff stage $k$ has a {\it discrete uniform} pdf (probability density function)  $f_k(\cdot)$ on the integers $\{ 0, 1, \cdots, 2b_k -1\}$  with mean ${1}/{q_k} = (2b_k -1)/2$ and variance $${\rm Var}_k = \frac{b_k^2}{3} - \frac{1}{12} = \frac{1}{3} \cdot \frac{2b_k+1}{2b_k -1} \cdot \frac{1}{q_k^2} = \frac{v_k^2}{q_k^2}$$ where the pre-factor is denoted by $v_k^2$ to simplify the exposition in the current section. Note that $\lim_{b_k \to \infty} v_k^2 = 1/3$. 


Let $\Omega$ and $f_{\Omega}(\cdot)$ respectively denote the sum of the backoff values generated for a packet, and its pdf. Also denote by $\bar \Omega$ its mean and $\sigma_{\Omega}^2$ its variance. 
It should be clear that the sum of the backoff values generated for a packet $\Omega$ which we baptize in this paper {\it per-packet backoff} can be formally defined as a {\it compound} random variable:
\begin{equation}\label{eq:defbsigma}
\setlength{\fboxrule}{1pt}\fcolorbox{magenta}{white}{$\textstyle \Omega \defeq \sum_{k=0}^\kappa B_k$ }
\end{equation}
where $B_{k}$ is a random variable denoting the backoff value generated at the $k$th backoff stage, for a packet of a tagged node, and $\kappa$ is also a random variable designating the highest backoff stage reached by the packet.

The probability that the $k$th backoff stage is reached during the backoff duration for a packet can be computed as $\gamma^k$ irrespective of the backoff distribution at any backoff stage. Hence we have $$\probability [ \kappa =k] = \gamma^k - \gamma^{k+1}, \quad \forall k \in \{0, \cdots,K-1\},$$ and $\probability [\kappa = K] = \gamma^K$. From Bayes' theorem, $f_\Omega(\cdot)$ becomes:
\begin{equation}
\textstyle f_{\Omega} (x) = \sum_{k=0}^K f_{\Omega} (x ~\vert ~ \kappa=k) \cdot \probability \left[ \kappa = k \right]
\end{equation}
where $f_{\Omega} (x ~\vert ~ \kappa=k)$ denotes the sum of the backoff values from $0$th to $k$th stages for a given $k$. Applying the fact that the sum of $k$ random variables with pdfs $f_0(\cdot),\cdots,f_k(\cdot)$ has a pdf of the convolution of the pdfs yields \begin{equation}\label{eq:totalpdf}
\textstyle f_{\Omega} (x) = f^{*K}(x) \gamma^{K} + (1-\gamma) \sum_{k=0}^{K-1} f^{*k}(x) \gamma^k
\end{equation}
where $f^{*k}(\cdot) \defeq ( f_0* \cdots *f_k )(\cdot)$ is the convolution of $k+1$ functions. In a similar way, $\bar \Omega$ can be computed from \eqref{eq:defbsigma}:
\begin{align}
\textstyle\bar \Omega& = \textstyle\expectation \left[ \sum_{k=0}^\kappa B_k \right] = \sum_{k=0}^K \expectation \left[ {\sum_{k'=0}^k B_{k'}}  \right] \cdot \probability[\kappa = k] \nonumber \\
& \textstyle= \sum_{k=0}^K \left( \sum_{k'=0}^k \frac{1}{q_{k'}} \right) \cdot \probability[\kappa = k]  \label{eq:thirdeqbarb}.
\end{align}
By manipulating \eqref{eq:thirdeqbarb} combined with the expression of $\probability[\kappa = k]$, it is easy to see that
\begin{align}
 \bar \Omega =   \sum_{k=0}^K \frac{\gamma^k}{q_k}. \label{eq:totalmean}
\end{align}
In addition, using $\expectation [  B_k^2 ] = ({1+v_k^2})/{q_k^2} $, the second moment of $\Omega$ can be rearranged as
\begin{align}
& \textstyle\sigma_{\Omega}^2 + \bar \Omega^2 = \expectation \left[ \left( \sum_{k=0}^\kappa B_k \right)^2 \right] \nonumber \\
& =  \textstyle  \sum_{k=0}^K \expectation \left[ \left( {\sum_{k'=0}^k B_{k'}} \right)^2   \right] \probability[\kappa = k] \label{eq:secondcov} \\
& =  \textstyle \sum_{k=0}^K \textstyle \left( \displaystyle\sum_{k'=0}^k \textstyle \frac{1+v_{k'}^2}{q_{k'}^2}  + 2 \displaystyle \sum_{i=1}^k \sum_{j=0}^{i-1} \textstyle \frac{1}{q_{i} q_{j}} \right) \probability[\kappa = k]    \label{eq:thirdcov} \\
& = \textstyle \left( \sum_{k=0}^K \frac{\gamma^k}{q_{k}^2} (1+v_k^2) \right)+ 2 \left( \sum_{k=1}^K \frac{\gamma^k}{q_{k} } \sum_{i=0}^{k-1} \frac{1}{q_{i}} \right) \nonumber
\end{align}
The above equalities can be easily verified by rearranging \eqref{eq:secondcov} and \eqref{eq:thirdcov}. Moreover, it is shown in Appendix \ref{sec:dertotal} that, if $q_{k} = 2 / (2 b_0 m^k -1)$ as in the standard, $v_{\Omega}^2 \defeq \sigma_{\Omega}^2/\bar \Omega^2$ simplifies to
\begin{align}
  \frac{ \textstyle \sum_{k=0}^K \delta_k }{ \left( \textstyle \sum_{k=0}^{K} (b_0 m^k - 1/2 ) \gamma^k \right)^2 } - 1 \label{eq:totalcovcor}
\end{align}
where $$\textstyle \delta_k =\left(b_0 m^k - \frac{1}{2} \right) \gamma^k \left\{ \left( \frac{m+1}{m-1} +v_k^2 \right) \left(b_0 m^k - \frac{1}{2} \right) - k - \frac{2b_0-1}{m-1}  \right\}.$$

\begin{remark}\label{rem:perpacket}

\ifthenelse{\boolean{longver}}{We spare our breath for later sections, and briefly present only main points.}{}

\begin{asparaenum}[\rmfamily\bfseries{R\ref{rem:perpacket}}.1]
\ifthenelse{\boolean{longver}}{ \item An astute reader might realize that $v_\Omega$ as shown in \eqref{eq:totalcovcor} plays a key role when we apply central limit theorem in later sections to compare the random sums of i.i.d. random variables.}{}
\item The result puts forward an {\it alternative} viewpoint. We can view the backoff process {\it reflecting} the collision effect among nodes as if there is {\it no collision} at all and the per-packet backoff for every node has a distribution with mean $\bar\Omega$ and CV $v_\Omega$ (or equivalently variance $\sigma_\Omega^2$).
\item {}[{Answer to \bfseries\textcolor{webmag}{Q1}}] Consider the case $N=2$. It can be computed from \eqref{eq:totalcovcor} that $\Omega$ is approximately {\it  uniformly} distributed in 802.11b while it is {\it exponentially} distributed in 802.11a/g in the sense that $v_\Omega \approx 0.7$ (though slightly larger than $1/\sqrt{3}$) and $v_\Omega \approx 1.0$, respectively, mainly due to different initial contention windows ($2b_0=32$ in 802.11b and $2b_0 = 16$ in 802.11a/g). \ifthenelse{\boolean{longver}}{Thus $\Omega$ in 802.11a/g can be deemed exponentially distributed for $N=2$.}{} This is the reason why they \cite{refBergerUniform,refMarkusExp} observed that their testbed data closely match the expressions of inter-transmission probability $\probability [Z\vert \zeta] $, which were derived under their respective assumptions. Note that we communicated with the first author of \cite{refMarkusExp} to verify the protocol (802.11g) used in their testbed. We will formally define $\probability [Z\vert \zeta] $ in Section \ref{sec:interprobintro}.
\end{asparaenum}
\end{remark}

To verify the analysis, simulations have been conducted. We have used {\it ns-2} version 2.33 with its built-in 802.11 module and the parameter set of 802.11b, \ie, $m=2$ and $2b_0 = 32$, except that $K$ is varied to observe the asymptotic property. All simulations use a $3000s$ warm-up period and all quantities are measured over the next $320,000s$ ($\approx90h$).

\begin{figure}[t!]
\centering
\centerline{\includegraphics[width=8.2cm]{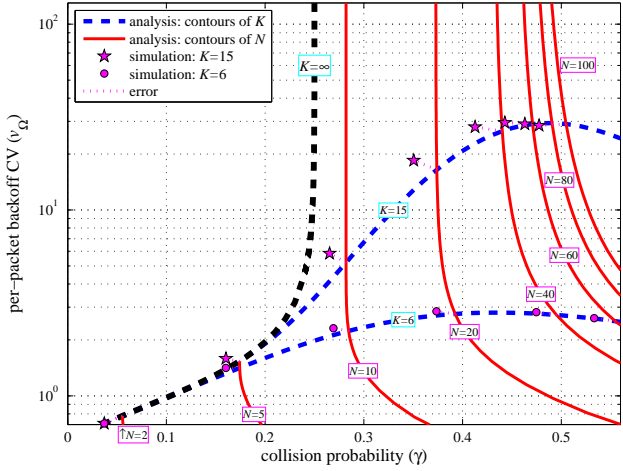}}
\caption{Per-packet backoff CV $v_\Omega$ vs. collision probability $\gamma$ for $K=6,15,\infty$; and $N=2,\cdots,100$.} \label{fig:perpacket}
\end{figure}

Fig. \ref{fig:perpacket} presents the per-packet CV $v_\Omega$, computed from \eqref{eq:totalcovcor}, \eqref{eq:fixedpoint} and \eqref{eq:fixedpointg}, and compared with the simulation results. The figure shows a good match between them. In the figure, the intersecting points of contours of $K$ and $N$ at each level decide $v_\Omega$ and $\gamma$ simultaneously. As is predicted by \eqref{eq:totalcovcor}, $v_\Omega$ goes to $\infty$ as $K$ goes to $\infty$ for $\gamma \geq 1/m^2 = 0.25$. It is remarkable that for a given $N \geq 9$ ($N \geq 5$ for 802.11a/g), $v_\Omega$ is extremely sensitive to $K$, forming a striking contrast with the insensitivity of $\gamma$ to $K$.

The discrepancy between analysis and simulation study is partly due to {\it reduced contention effect}, which is a less-known subtlety of DCF behavior discovered by Bianchi \etal \cite{refReduced} and is shown through simulations to be a factor of error by Sakurai and Vu \cite{refSakuraiDelay}.

\section{Point Process Approach: Poissonian Insights}\label{sec:intertxdist}

A basic property of per-packet backoff $\Omega$ discovered by Kwak \etal \cite[Theorem 1]{refKwak} and later strengthened by Kumar \etal \cite[Theorem 7.2]{refKumar07} is that the mean of per-packet backoff is proportional to the population, \ie, $\bar\Omega = \Theta(N)$. This turns out to play a key role in our point process approach in this section.

\subsection{Justification of Point Process Approach}\label{sec:ppapp}
In order to justify our point process approach, we need to show that the backoff process of each node has nonzero intensity, \ie, $\bar\Omega = \expectation[\Omega]$ is finite. Though, for finite $K$, this is self-evident from the form of \eqref{eq:totalmean}, we need to assume the following to prove $\bar\Omega <\infty$ for $K=\infty$.
\begin{compactenum}[\bfseries{A}.\arabic{acounter}]
\item\addtocounter{acounter}{1} $q_{k} = 2 / (2 b_0 m^k -1)$ for all $k \in \{ 0,\cdots,K\}$, and $m>1$.
\end{compactenum}
Under this assumption we can prove the following lemma which assures us that $\bar\Omega$ is finite whether $K$ is finite or not. We also would like to point out that a part of the proof of \cite[Theorem 7.2]{refKumar07}, which corresponds to the case $K=\infty$ of Lemma \ref{lem:meanexist} in our work, has a flaw because they should have proven $\gamma < 1/m$ before using $\sum_{k=0}^{\infty} (m\gamma)^k = 1/(1-m\gamma)$. 
\begin{lem}[Mean Exists]\label{lem:meanexist} \nextline
Under the above assumption, there exists a finite $K_0$ such that $\gamma <1/m$ and $\gamma$ is decreasing in $K$. This implies:
 \begin{compactitem}
 \item there exist $K_0$ such that $\gamma <1/m$ for all $K \geq K_0$ including $K=\infty$,
 \item the mean $\bar \Omega = \expectation[ \Omega ]$ exists for $K=\infty$.
 \end{compactitem}
\end{lem}
\begin{proof}
Suppose $\gamma \geq 1/{m}$. Then we have from \eqref{eq:fixedpoint} and $q_{k} = 2 / (2 b_0 m^k -1)$ that, for any $\epsilon>0 $, there exists $K_1$ such that $\bar p < \epsilon$ for all $K \geq K_1 $. In the meantime, from $1 - {\rm e}^{-x} \leq x$, we also have $\gamma \leq (N-1)\bar p < (N-1)\epsilon$. This contradicts $\gamma \geq {1}/{m}$, implying that there must exist $K_0$ such that $\gamma < {1}/{m}$ for $K=K_0$.

Denote the right-hand side of \eqref{eq:fixedpoint} by $\mathbf{P}(K)$. Since the right-hand side of \eqref{eq:fixedpointg} is increasing in $\bar p$ and $\mathbf{P}(K)$ is nonincreasing in $\gamma$ from \cite[Lemma 5.1]{refKumar07}, $1 - {\rm e}^{- (N-1){\mathbf{P}(K)}}$ is nonincreasing in $\gamma$. Therefore, it {\it suffices} to show that $$1 - {\rm e}^{- (N-1){\mathbf{P}(K_0+1)}} \leq 1 - {\rm e}^{- (N-1){\mathbf{P}(K_0)}},$$ or equivalently $\mathbf{P}(K_0 +1) \leq \mathbf{P}(K_0)$, for all $\gamma < 1/m$. After some manipulation and some intricate factorization, it can be verified that $\mathbf{P}(K_0) - \mathbf{P}(K_0+1)$ takes the form:
\begin{equation}\nonumber
 \frac{b_0 \gamma^{K_0+1}  \sum_{k=0}^{K_0} \gamma^{k} \left( m^{K_0+1} - m^k  \right) }{\left\{ \sum_{k=0}^{K_0} \left( b_0 m^k-\frac{1}{2}\right) \gamma^k \right\} \left\{ \sum_{k=0}^{K_0+1} \left( b_0 m^k-\frac{1}{2}\right) \gamma^k \right\}}
\end{equation}
which is greater than zero for $m>1$, implying that the solution $\gamma^*$ of \eqref{eq:fixedpoint} and \eqref{eq:fixedpointg} for $K=K_0+1 $ is smaller than that for $K=K_0$.
Applying mathematical induction completes the proof. Also note that this implies $\gamma < 1/m$ for any $K \geq K_0$.

For the case $K=\infty$, since we have shown that $\gamma<1/m$ is decreasing in $K$ for all $K \geq K_0$, it follows from \cite[Theorem 3.14]{refAnalysis} that as $K$ goes to infinity, $\gamma$ should converge to $\hat\gamma < 1/m$. The existence of $\expectation[ \Omega]$ follows from \eqref{eq:totalmean}.
\end{proof}

Since $m>1$ guarantees that there exists $K_0$ such that $\gamma < 1/m$ for all $K\geq K_0$, it can be seen from \eqref{eq:phiknew} that, for the case of $K=\infty$, $ m>1 $ is also a sufficient condition for the existence of $\kappa$ such that $\phi_k > \phi_{k+1}$ for all $k \geq \kappa $, \ie, the average number of nodes in backoff stage $k$ is larger than that in backoff stage $k+1$. This corresponds to the tightness condition of $\phi_k$, which prevents a node from escaping to infinite backoff stage \cite{refBillConvergence}. 
The fact that the condition $m>1$ prevents a node from escaping to infinite backoff stage appears to be in best agreement with our usual intuition.

\subsection{Essential Assumption}\label{sec:essass}
To establish Poisson limit result in Theorem \ref{th:poi} and to justify point process approach in the remaining sections, we need the following essential assumption.
\begin{compactenum}[\bfseries{A}.\arabic{acounter}]
\item\addtocounter{acounter}{1} Per-stage backoff distribution $f_k(\cdot)$ is a {\it uniform} {\it continuous} function. It also means $v_k=1/\sqrt3$.
\end{compactenum}
Recall that $f_\Omega(\cdot)$ is expressed by \eqref{eq:totalpdf}, hence now it is a weighted sum of convolutions of continuous pdfs $f_k(\cdot)$ where the weight for each convolution function $f^{*k}(\cdot) = ( f_0* \cdots *f_k )(\cdot)$ is $(1-\gamma) \gamma^k$, which is a function of $\gamma$. As we noted in Remark \ref{rem:perpacket}, $f_\Omega(\cdot)$ reflects the collision effect through $\gamma$ which determines how much $ f_\Omega (\cdot)$ is dispersed.

\noindent {\bf On continuity assumption}: Denote by $D^n(t)$ the number of cumulative per-node successful transmissions until time-slot $t$. Formally, $D^n(t)$ is {\it discrete-time} renewal process that counts the number of arrivals during the interval $[0,t]$ where the inter-arrival times are i.i.d. copies of discrete random variable $\Omega$. Consider superposition process $D(t) \defeq \sum_{n=1}^N D^n(t)$. A {\it subtlety} in 802.11 is that there may be no intervening backoff time-slot between two consecutive successful transmissions. More precisely, at the beginning of a backoff time-slot, if the transmission attempts of nodes lead to a successful transmission, the time-slot is rendered unused, meaning that the time-slot is reused after the successful transmission. The same subtlety applies to collision events. Simply suppose the probability that a successful transmission (or a collision event) occurs at the beginning of a time-slot converges to $P_S$ (or $P_C$) as $N\to\infty$. Putting $$\textstyle P(x)\defeq  \probability[\lim_{N\to\infty} D(t+1)-D(t)=x],~~x\in\{0,1,\cdots\},$$
we can see from the subtlety that
\begin{align}
\textstyle P(x+1) = P(x) \cdot \sum_{i=0}^\infty P_C^i P_S= \frac{P_S}{1- P_C} P(x) . \nonumber
\end{align}
Because $\sum_{x=0}^{\infty} P(x) = 1$, we have a geometric distribution
\begin{align}
\textstyle P(x)  =  \left(1 - \frac{P_S}{1- P_C} \right) \left( \frac{P_S}{1- P_C} \right)^x ,~~x\in\{0,1,\cdots\} \nonumber
\end{align}
hence the limiting (as $N\to\infty$) distribution of cumulative process $D(t)$ for arbitrary integer $t$ takes a {\bf Pascal} (negative binomial) distribution\footnotemark\footnotetext{Sakurai and Vu \cite[Section III-B]{refSakuraiDelay} assumed $D(t)$ is a Bernoulli process. This simplification was justified by the reduced contention effect \cite{refReduced}.}. This fact can be exploited for a more accurate approximation. A simpler approximation at the cost of accuracy is to be presented in Theorem \ref{th:poi}.


Once again, the continuity assumption turns out unavoidable in Section \ref{sec:asymptotic} because regular variation theory \cite{refRegularBingham} exploited by Theorem \ref{th:powertail} is not well developed for discrete functions.
The uniform distribution assumption of $f_k(\cdot)$ was made only to simplify the exposition of Theorems \ref{th:heavy} and \ref{th:powertail} in Section \ref{sec:asymptotic}.

\subsection{Poisson Process Approximation}\label{sec:poissonapprox}

We can now view the backoff procedure of node $n$ as a stationary simple {\it renewal process} $A^n(t)$ that counts the number of arrivals during the interval $(0,t]$ where the $j$th inter-arrival times, $T_j^n - T_{j-1}^n$, are given by the i.i.d. copies of the continuous random variable $\Omega$. Then the backoff procedure of all nodes can be regarded as a {\it superposition} of $N$ statistically identical renewal processes, \ie,
\begin{equation}
\textstyle A(t) \defeq\sum_{n=1}^{N} A^n(t). \nonumber
\end{equation}
It should be remarked that, if one or more component processes are not Poisson, the superposition process $A(t)$ is {\it not renewal}, and even if the inter-arrival times of $A(t)$ are identically distributed, they are {\it not independent} \cite{refAlbinOR}.

In the following, we present a novel way to tackle this analytical difficulty caused by the dependence among the inter-arrival times of the superposition process. The key observation is that the {\it entropy} of the superposition point process $A(t)$ increases with $N$, which is implied by the following known result \cite[Proposition 11.2.VI]{refPointProcess}.

\begin{lem}[Poisson Limit for Superposition]  \label{lem:poissonlimit}  \nextline
Let $\Xi(t)$ denote the point process obtained by superposing $M$ {\rmfamily\bfseries independent} replicates $B^m(t)$, $m\in\{1,\cdots,m\}$, of a simple
stationary point process  with intensity $\lambda$ and dilating the time-scale by a factor $M$. Formally speaking,
\begin{align}\label{eq:suplimit}
\textstyle \Xi(t) = \sum_{m=1}^M B^m(t/M).
\end{align} Then as $M \to \infty$, $\Xi(t)$ converges weakly to a Poisson process with the intensity $\lambda$.
\end{lem}


Now it follows from the basic property \cite[Theorem 7.2]{refKumar07} for $K=\infty$ that the mean inter-arrival time of $A^n(t)$, $\bar \Omega$, is of order $N$. Therefore, there must exist a point process $$B^n(t) \defeq \lim_{N\to\infty} A^n(Nt) ~\mbox{with {intensity} }\lambda = \lim_{N\to\infty} N/\bar\Omega$$ where intensity $\lambda$ does {\bf not scale} with $N$ and we have $B^n(t/N) \approx A^n(t)$ as $N$ goes to $\infty$. This in turn implies
$$ \textstyle  \sum_{n=1}^{N} A^n(t )  ~\approx~  \sum_{n=1}^N B^n(t/N)
$$
which has the same form of \eqref{eq:suplimit}. Applying Lemma \ref{lem:poissonlimit} to the above equation leads to the following theorem.
\begin{theorem}[Dichotomy of Aggregation: First Part] \label{th:poi}  \nextline
Suppose $\bar\Omega=\Theta(N)$. Then the superposition process $\sum_{n=1}^{N} A^n(t )$ converges weakly to a Poisson process as $N\to\infty$.
\end{theorem}
\begin{remark}\label{rem:poiss}
 This result states that the Poissonian nature is inherent in the backoff process of 802.11 and provides an answer to {\bfseries\textcolor{webmag}{Q3}}.
\begin{asparaenum}[\rmfamily\bfseries{R\ref{rem:poiss}}.1]
\item {\bf The reason we do not require ${K=\infty}$}: Recalling our discussion at the beginning of this section, we can see that
    $$ K=\infty ~\stackrel{\scriptstyle\mbox{\footnotesize\cite[Theorem 7.2]{refKumar07}}}{\Longrightarrow}~ \bar\Omega = \Theta(N) ~\stackrel{\scriptstyle\mbox{\footnotesize{Theorem \ref{th:poi}}}}{\Longrightarrow}~ \mbox{Poisson}.$$
    If we require $K=\infty$ instead of $\bar\Omega = \Theta(N)$, the above theorem would look simpler, but it would not be applicable for the case $K<\infty$. Even if $K$ is {\bf finite}, the crucial scaling condition $\bar\Omega = \Theta(N)$ holds for a wide range of $N$, as hinted by previous works (See the simulation result with a practical parameter set in \cite[Figures 2 and 5]{refSakuraiDelay}). However, for extremely large $N$, the scaling becomes $\bar\Omega = \Theta(1)$. 
\item From a different angle, the backoff procedure of 802.11 along with its setting $K=6$ is intentionally designed so that the {\it successful} attempt intensity of each node $1/\bar\Omega$ is kept being of the order of $1/N$ for a wide range of $N$, by allowing enough number of backoffs for each packet.
\end{asparaenum}
\end{remark}

\noindent {\bf What is the premise of Poisson limit?}: The question remains whether the approximation is precise even for $t=\infty$. As Whitt discussed in \cite[Chapter 9.8]{refWhitt}, the underlying assumption of the Poisson limit theorem (Lemma \ref{lem:poissonlimit}) is that $t$ is finite. In the meantime, the basic premise of the Poisson limit theorem is that the component process $A^n(t)$ should become sparse ($\bar\Omega = \Theta(N)$) \cite[pp.83]{refWhittSuper}. If we allow $t\to\infty$ at the same time as $N\to\infty$, $A^n(t)$ may not remain sparse. This is essentially why we {\it must} adopt an another approximation in Section \ref{sec:shortterm} where $t=\Theta(N)$. In the light of these points, the above theorem provides a natural approximation of the backoff processes on normal time-scale, as compared with the other approximation in Section \ref{sec:shortterm} on coarse time-scales.

\section{Asymptotic Analysis}\label{sec:asymptotic}

A stochastic process with infinite variance and self-similarity exhibits phenomena called {\it Noah effect} and {\it Joseph effect}, respectively, in Mandelbrot's terminology \cite{refWhitt,refFBMTaqqu}. Noah and Joseph effects refer to the biblical figures Noah, who experienced an extreme flood -- exceptionally large values -- and, Joseph, who experienced long periods of plenty and famine -- self-similarity or strong positive dependence. This section lifts the veil to discover these effects and to explain their influences on the backoff process in 802.11. We have not assumed $K=\infty$ because all results derived so far are applicable if either of finite and infinite $K$ is used (See Remark \ref{rem:poiss} also). However, all results derived in this section require $K=\infty$, hence we formally assume the following.
\begin{compactenum}[\bfseries{A}.\arabic{acounter}]
\item\addtocounter{acounter}{1} There are {\bf infinite} backoff stages, \ie, $K=\infty$.
\end{compactenum}

\subsection{Moment Analysis}

We introduce the notion of a wide-sense heavy-tailed distribution borrowed from \cite{refHeavyRolski}. We call a pdf $f(x)$ {\it wide-sense heavy-tailed} if its moment generating function is infinite, \ie, $$\textstyle \int_{0}^{\infty} {\rm e}^{tx} f(x) \ud x = \infty, ~ \forall t>0.$$
We now characterize the existence of all fractional moments of $\Omega$. Let us define $$\setlength{\fboxrule}{1pt}\fcolorbox{magenta}{white}{$\alpha \defeq - (\log\gamma)/{\log m}$}$$
  where $\alpha>1$ is satisfied by Lemma \ref{lem:meanexist}. Also it is remarkable that Sakurai and Vu \cite{refSakuraiDelay} established a similar result for integer moments. Note however that we {\it cannot} prove Theorem \ref{th:powertail} without the following extended result for fractional moments.

\begin{theorem}[Existence of Fractional Moments] \label{th:heavy}  \nextline
The per-packet backoff $\Omega$ has a wide-sense heavy-tailed distribution. In addition, its $c$th moment $\expectation[ \Omega^c ]$ is
 \begin{compactitem}
 \item
  {\bfseries infinite} if $c \geq \alpha $,
 \item and {\bfseries finite} if $0 \leq c < \alpha $.
 \end{compactitem}
\end{theorem}
\begin{proof}
First we note $\alpha = - (\log\gamma)/{\log m}$ is equivalent to $m^\alpha \gamma = 1$. It also follows from Lemma \ref{lem:meanexist} that $\alpha>1$. Letting $c$ be any real number such that $c \geq \alpha$, we have $m^c \gamma \geq 1$. Then the $c$th moment of $\Omega$, $\expectation[ \Omega^c ]$, can be computed as
\begin{align}
\textstyle  & \textstyle\sum_{k=0}^{\infty} \expectation \left[ \left( {\sum_{k'=0}^k B_{k'}} \right)^c   \right] \cdot \probability[\kappa = k] \nonumber \\
& \textstyle\geq \sum_{k=0}^{\infty} \left( \expectation \left[ {\sum_{k'=0}^k B_{k'}}  \right] \right)^c \cdot \probability[\kappa = k] \nonumber \\
& \textstyle= \sum_{k=0}^{\infty} \left( \sum_{k'=0}^k (b_0 m^{k'}-\frac{1}{2}) \right)^c \cdot \probability[\kappa = k]  \nonumber\\
& \textstyle\geq \sum_{k=0}^{\infty} \sum_{k'=0}^k (b_0 m^{k'}-\frac{1}{2})^c   \cdot \probability[\kappa = k] \nonumber \\ & \textstyle = \sum_{k=0}^{\infty} (b_0 m^{k}-\frac{1}{2})^c  \gamma^k \nonumber
\end{align}
where the first inequality holds by H\"older's inequality for expectations, \ie, $(\expectation[X])^c \leq \expectation[X^c]$, and the second inequality follows from $c > 1$. Hence, from the last expression, we have $\expectation[ \Omega^c ] \to\infty$ as $K \to \infty$. Note that $c$ is real. Since there exist infinite moments, $\Omega$ has a wide-sense heavy-tailed distribution. Now consider the $c$th moment for $1<c<\alpha$.
\begin{align}
& \textstyle\expectation \left[ \left( \sum_{k=0}^\kappa B_k \right)^c \right]  = \sum_{k=0}^{\infty} \expectation \left[ \left( {\sum_{k'=0}^k B_{k'}} \right)^c  \right] \cdot \probability[\kappa = k] \nonumber \\
& \leq \textstyle\sum_{k=0}^{\infty} \expectation \left[ (k+1)^{c-1} \sum_{k'=0}^k \left(  B_{k'} \right)^c  \right] \cdot \probability[\kappa = k] \label{eq:suffitemp1} \\
& = \textstyle\sum_{k=0}^{\infty}  (k+1)^{c-1} \sum_{k'=0}^k \frac{(2b_0 m^{k'}-1)^c}{(c+1)} \cdot \probability[\kappa = k] \label{eq:suffitemp2} \\
& \leq \textstyle\frac{(2b_0)^c}{(c+1)} \sum_{k=0}^{\infty}  (k+1)^{c-1} \frac{(m^c)^{k+1} - 1}{m^c-1} \cdot \probability[\kappa = k] \nonumber \\
& \leq \textstyle\frac{(2b_0)^c}{(c+1)}  \sum_{k=0}^{\infty}  (k+1)^{c-1} \frac{(m^c)^{k+1} \gamma^k}{m^c-1} \label{eq:suffitemp3} \\
& = \textstyle\frac{(2b_0 m)^c}{(c+1)(m^c-1)} \sum_{k=0}^{\infty}  (k+1)^{c-1} (m^c\gamma)^{k} \label{eq:suffitemp4}
\end{align}
where \eqref{eq:suffitemp1} can be obtained by applying original H\"older's inequality, \ie, $$\textstyle\left(\sum_{k'=0}^k 1\cdot b_{k'} \right) \leq \left(\sum_{k'=0}^k 1^{\frac{c}{c-1}}\right)^{\frac{c-1}{c}} \left( \sum_{k'=0}^k (b_{k'})^c \right)^{\frac{1}{c}}. $$ \eqref{eq:suffitemp2} can be verified by computing $\int b^c f_{k'}(b) \ud b$ where $f_{k'}(b)$ is a uniform pdf with mean $b_0 m^{k'} - 1/2$. \eqref{eq:suffitemp3} follows from $\probability[\kappa =k] \leq \gamma^k$. Then it suffices to show that d'Alembert's ratio of the series \eqref{eq:suffitemp4} is less than one. Recalling that $m^c \gamma < m^\alpha \gamma = 1$, we can see that
\begin{equation}
\lim_{k \to \infty} \frac{(k+2)^{c-1} (m^c \gamma)^{k+1}}{(k+1)^{c-1} (m^c \gamma)^{k}} = m^c \gamma < 1 \nonumber.
\end{equation}
This establishes \eqref{eq:suffitemp4} is finite for $K=\infty$, and completes the proof.
\end{proof}
 \ifthenelse{\boolean{longver}}{A bright spot in the misfortune is that $\alpha>1$ is guaranteed thanks to Lemma \ref{lem:meanexist} so that $\bar\Omega$ is always finite.}{}
\begin{remark}\label{remark:heavy}
[{Answer to \bfseries\textcolor{webmag}{Q4}}]
 This theorem reveals that $\Omega$ is wide-sense heavy-tailed in the sense that {\it not} all of its moments exist, as Sakurai and Vu \cite[Theorem 1]{refSakuraiDelay} first noted.  \ifthenelse{\boolean{longver}}{The {\it necessary} and {\it sufficient} condition for the existence of the moments of $\Omega$ paves the way for the role of the constant $\alpha =  -(\log\gamma)/{\log m}$ as a ramification point.}{}
\end{remark}
 As shown in Fig. \ref{fig:perpacket}, the variance $\sigma_\Omega^2$ in 802.11b is not very large\ifthenelse{\boolean{longver}}{($\leq (3 \bar\Omega)^2 $).}{.} Nevertheless, the statistics of $\Omega$ certainly contain precursors of infinite-variance distributions, as shown in the next section.

\subsection{Strict-Sense Heavy-Tailedness: Tauberian Insights}

Although there has been some work to prove the wide-sense heavy-tailedness of the delay or backoff duration \cite{refSakuraiDelay} and the power-law like behavior of access delays was {\it identified} only through simulations in a few works \cite{refSakuraiDelay,refTickooSS}, to the best of our knowledge, {\it none} of them proved that the delay or backoff duration has a power-law tail. This quite intuitive property has not been established mainly due to the theoretical difficulties underlining the proof. It is important to note that this theorem is a prerequisite for mathematical analysis of Noah effect, which implies strict-sense heavy-tailedness.

We would like to place particular emphasis on the following theorem for another reason. We note that some work \cite{refChaoticVeres,refFigTCPSelf} considered the question whether a {\bf single} long-lived TCP flow can generate traffic that exhibits long-range dependence (or, equivalently, asymptotical second-order self-similarity). It is significant that long-range dependence is a property which is {\it automatically} implied by heavy-tailed inter-arrival times \cite{refLipsky} for the single flow (or node) case, irrespective of the context. That is, even a renewal process (no correlation of inter-arrival times) with heavy-tail distributed inter-arrival times generates long-range dependence in the counting process.
In the light of this point, one do not need to conduct analyses of tremendous traffic traces if there is a solid mathematical work that can settle this kind of dispute.

In the following theorem, we prove that the per-packet backoff distribution has a power tail by lighting upon the fact that the moment generating function has a {\it recursive relation}, and by applying the theory of {\it regular variation} \cite{refRegularBingham} and the less-known {\it modified Tauberian theorem} of Bingham \& Doney \cite{refTauberian}. For your own good, note that this theorem requires only $K=\infty$, nothing about $N$.

\begin{theorem}[Power Tail Principle{\footnotemark}]\label{th:powertail}  \nextline The per-packet backoff $\Omega$ has a Pareto-type tail with an exponent of $-\alpha$. Formally,
\footnotetext{The proof in fact requires $\alpha$ to be not an integer. For the complicated case when $\alpha$ is an integer, we refer to \cite{refMeyer} and \cite[Theorem 8.1.6]{refRegularBingham}. However, since an integer $\alpha$ can be approximated for any small $\epsilon>0$ by a real number $\tilde\alpha$ such that $|\alpha - \tilde\alpha| < \epsilon$, we expect the result of Theorem \ref{th:powertail} to be valid for all $\alpha>0$.}
 \begin{equation}\setlength{\fboxrule}{1pt}
\fcolorbox{blue}{white}{$\displaystyle F_{\Omega}^c (x) \defeq \displaystyle \int_{x}^{\infty} f_{\Omega} (x) dx ~\sim~ x^{-\alpha}  \slow\left( x \right)$}. \label{eq:ccdfregular}
\end{equation}
The notation $f(x) \sim g(x)$ means $\lim_{x\to\infty} {f(x)}/{g(x)} = 1$, and $\slow\left( x \right)$ is slowly varying\footnotemark.
\end{theorem}
\footnotetext{A function $f(x)$ is called {\it regularly varying} \cite{refRegularBingham} at infinity of index $\rho$ iff $\lim_{x\to\infty} {f(\lambda x)}/{f(x)} = \lambda^\rho, \forall \lambda>0$. For the special case $\rho =0$, it is called {\it slowly varying} and usually denoted by $\slow(x)$. For example, a positive constant, $(\log x)^\epsilon$ for any real number $\epsilon$ is a slowly varying function. A slowly varying function $\slow(x)$ is dominated by any positive power function, \ie, $\lim_{x\to\infty} {\slow(x)}/{x^\epsilon}=0$, $\forall \epsilon>0$.}
\begin{remark}\label{rem:powertail}
This principle, formulated in terms of the ccdf $F_{\Omega}^c (\cdot)$, not only defines a fundamental characteristic of delay but also lays the groundwork for further analysis using regular variation theory.
\begin{asparaenum}[\rmfamily\bfseries{R\ref{rem:powertail}}.1]
\item {} [{Answer to \bfseries\textcolor{webmag}{Q4}}] This clear-cut and simple result reveals the statistical attribute of  $\Omega$ for any population $N$. It has a Pareto-type distribution whose exponent parameter is $-\alpha$. Theorem \ref{th:powertail} proves the {\bf strict-sense heavy-tailedness} of $\Omega$ for $\alpha < 2$, and puts an end to the discussions in Section \ref{sec:intro}.
\item This theorem dispenses the {\it complicated} convolution expression \eqref{eq:totalpdf} and leads us to a {\it simpler} conclusion. The most representative distribution of backoff times $\Omega$ is a truncated Pareto-type distribution (though it must be slowly-varying), rather than uniform or exponential as observed in the simulation studies of \cite{refBergerUniform, refMarkusExp}.
\item The simplistic term $\slow(\cdot)$ in \eqref{eq:ccdfregular} is {\it irreplaceable} with any other expressions, implying its pivotal role. For instance, Final Value Theorem tells nothing but $\lim_{x\to\infty} f_\Omega (x) = 0$.
\end{asparaenum}
\end{remark}

\begin{figure}[t!]
\centering
\centerline{\includegraphics[width=8cm]{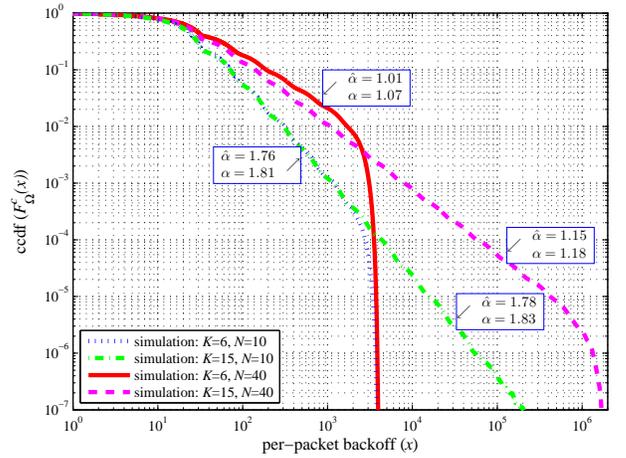}}
\caption{Complementary cumulative distribution function $F_\Omega^c (x) $ for $K=6,15$; and $N=10,40$.} \label{fig:power}
\end{figure}

The ccdf of $\Omega$ obtained through ns-2 simulations is plotted in Fig. \ref{fig:power} on a log-log scale where the estimated slopes $\hat\alpha$ are compared with the analytical formulae $\alpha = -(\log \gamma)/\log m$, \eqref{eq:fixedpoint} and \eqref{eq:fixedpointg}. Observe that these simple formulae along with \eqref{eq:ccdfregular} provide a precise estimate for the tail distribution. Remarkably, {\it even} for $K=6$, \ie, the value adopted in 802.11b, the ccdf of $\Omega$ can be accurately approximated by a truncated power-law tail.

\section{Short-Term Fairness Analysis}\label{sec:shortterm}

First of all, we cancel the assumption $K=\infty$ we made in Section \ref{sec:asymptotic} because we present in this section a new approximation for the superposition process and short-term fairness analysis, both of which will be applicable to both cases $K<\infty$ and $K=\infty$.

\subsection{Inter-Transmission Probability}\label{sec:interprobintro}

The notion of {\it short-term fairness} \cite{refKoksal,refBergerUniform,refMarkusExp}, defined as the distribution of successful transmissions of nodes for a {\it finite} time, has been getting the limelight due to its central role in quantifying the behavior of random access protocols over short time-scales and its close link to access delays.
Among the set of nodes $ \{ 1,\cdots,N \}$, we tag node $N$, without loss of generality. Assume that the tagged node successfully transmitted a packet at time $t=0$. Denote by $Z_n$ the number of packets successfully transmitted by node $n$ {\it while} the tagged node transmits $\zeta$ packets. Recalling that $A^n(t)$ counts the arrivals during the interval $(0,t]$, we can see $$Z_n \defeq A^n( t' ) ~\mbox{where}~ t' = \min \{t:A^N(t) =\zeta \} .$$ It is clear that $Z_N = \zeta$ from the above definition. For short-term fairness analysis, we consider $$ \setlength{\fboxrule}{1pt}\fcolorbox{magenta}{white}{$Z= \sum_{n=1}^{N-1} Z_n $ } . $$ For the sake of convenience, we denote $\probability [Z=z \vert Z_N = \zeta]$ by $\probability_N [z \vert \zeta]$. We call the conditional probability $\probability_N [z \vert \zeta]$ {\it inter-transmission probability}. In terms of the point processes $A^n(t)$, it is equivalent to
\begin{align}
\textstyle \probability_N [z \vert \zeta] = \probability \left[ \sum_{n=1}^{N-1} A^n\left(   \sum_{j=1}^{\zeta} \Omega_j \right)  = z \right] \nonumber
\end{align}
where $\Omega_j$ denotes the per-packet backoff for each $j$th packet of the tagged node $N$ and are i.i.d. copies of $\Omega$.

\subsection{Intermediate Telecom Process on Coarse Time Scales}\label{sec:itpapprox}
\noindent {\bf The premise does not hold}: Look into the above superposition process $\sum_{n=1}^{N-1} A^n (   t )$ where $t=\sum_{j=1}^{\zeta} \Omega_j$. Recall the basic premise of Poisson limit theorem (Lemma \ref{lem:poissonlimit}) is that each component process must become sparse as $N$ grows. It is easy to see that this premise does not hold any longer here because $t= \sum_{j=1}^{\zeta} \Omega_j $ is of order of $\zeta\cdot \bar\Omega$ in the sense that $\expectation[t] = \Theta(\zeta \bar\Omega)$ and $\bar\Omega $ is of order of $N$ in most cases (See Remark \ref{rem:poiss}). Therefore, we need a new approximation of the superposition process on {\bf coarse time-scales} such that $t = \Theta (N)$.

Before that, we epitomize {\it theory of stable law} \cite[Chapter 4]{refWhitt} briefly only for the case $\alpha \in (1,2]$. Denote by $\mathbf{S}_\alpha(\sigma,\beta,\mu)$ {\it L\'evy $\alpha$-stable laws} whose four parameters are: the {\it index} $\alpha$; the {\it scale} parameter $\sigma$; the {\it skewness} parameter $\beta$; and the mean $\mu$. If $X_1,\cdots,X_n$ are i.i.d. copies of $\mathbf{S}_\alpha(\sigma,\beta,\mu)$, they satisfy the {\it stability} property which takes the following form
\begin{equation}
\textstyle \sum_{i=1}^{m} (X_i  - \mu) \disteq m^{\frac{1}{\alpha}} (X_1 - \mu) \nonumber
\end{equation}
where the notation $\disteq$ means equality in distribution. The case $\alpha =2$ is singular because we have $\mathbf{S}_2(\sigma,\beta,\mu) = \mathbf{N}(\mu,2\sigma^2)$ where $\beta$ plays no role. However, for the rest of cases $\alpha \in (1,2)$, there is no closed form expression for its pdf.

 Since Leland \etal \cite{refLeland} created a wave of interest in the self-similarity in the Internet, the probabilistic community has been concerned with the limit processes of aggregate renewal processes under different limit regimes. Here a point at issue was the order of limit operations, \ie, $t\to\infty$ and $N\to\infty$. Recently, Kaj \etal \cite{refPoissonBridge,refLimitFractal,refConvFBM} have established a fundamental connection between Noah effect and Joseph effect, elucidating the above issue as well.

 \noindent {\bf Aggregate Process on Coarse Time Scales}: A premise of \cite[Theorem 1]{refLimitFractal} is that each component process should {\bf not} become sparse as $N$ grows, \ie, inter-arrival times not scaling with $N$. This premise is fully satisfied when we consider $A^n(\bar\Omega t)$ instead of $A^n(t)$. In other words, we now view $A^n(\tau)$ on coarse time-scales $\tau = \bar\Omega t$. Also note that $\expectation [A^n(\bar\Omega t)] = t$. Then applying \cite[Theorem 1]{refLimitFractal} yields to the following result which is applicable to various cases $K=\infty$, $K<\infty$, finite time (which must be large enough though), and infinite time.

 \begin{theorem}[Dichotomy of Aggregation: Second Part{\footnotemark}]\label{th:itp}
Suppose, for $K=\infty$, the inter-arrival times of $A^n(\bar\Omega t)$ has ccdf $F_{\Omega}^c (\bar\Omega x) $ in \eqref{eq:ccdfregular} \underline{which does not vary with $N$}. For $K<\infty$, nothing is assumed. Define the centred superposition process $$ \textstyle \widetilde{A}(t) \defeq \left\{ \sum_{n=1}^{N} A^n( \zeta\bar\Omega t ) \right\} - N \zeta t . $$
Then, as $\zeta \to \infty$ and $N \to \infty$, we have
 \begin{align}
 \frac{\widetilde{A}(t)}{\zeta}   & \stackrel{\mbox{\footnotesize{weakly}}}{\longrightarrow}  -c \cdot \mathbf{Y}_{\alpha} \left( \frac{ t}{c} \right), ~ \mbox{for}~K=\infty, ~\alpha\in(1,2), \label{eq:itpinfinity} \\
  \frac{\widetilde{A}(t)}{\sqrt{N \zeta}}  &  \stackrel{\mbox{\footnotesize{weakly}}}{\longrightarrow}  v_\Omega \cdot \mathbf{B}( t), ~   \left\{ \begin{array}{l} \mbox{for}~K<\infty,  \\ \mbox{for}~ K=\infty,~\alpha\in(2,\infty), \end{array} \right. \label{eq:itpfinite}
 \end{align}
 where the scaling constant $ c \defeq \{ N\bar\Omega^{-\alpha} \slow(\zeta\bar\Omega) \}^{1/(\alpha-1)} /\zeta $, $\mathbf{B}(\cdot)$ is a standard Brownian motion, and $\mathbf{Y}_{\alpha} (\cdot)$ belongs to the family of {\it Intermediate Telecom process} \cite{refConvFBM} of index $\alpha$ whose cgf takes the form
\begin{align}
& \textstyle \log \expectation \left[ {\rm e}^{\theta \mathbf{Y}_\alpha (\tau)} \right]   \textstyle = \frac{\tau^{1-\alpha}}{\alpha-1} \left( {\rm e}^{\theta \tau } - 1- \theta \tau \right)  \nonumber \\ & \textstyle + \int_0^\tau \left( {\rm e}^{\theta x } - 1- \theta x \right) \left( \alpha \tau x^{-\alpha-1} +(2-\alpha) x^{-\alpha} \right) \ud x . \label{eq:cgf}
\end{align}
\end{theorem}
\footnotetext{Consistency between \eqref{eq:itpfinite} and Theorem \ref{th:poi}: Suppose $K=\infty$ and $\alpha \in (2, \infty)$ (which is very unlikely as $N$ must be large). Then assume the superposition process $A(\zeta\bar\Omega t)$ is {\bf Poisson}. For large $\zeta \bar\Omega$, this Poisson process should have a Gaussian marginal distribution with mean $ N \zeta t$ and variance $N \zeta t$, whereas the process \eqref{eq:itpfinite} has mean $N\zeta t$ and variance $v_\Omega^2 N \zeta t$. Therefore, Theorem \ref{th:poi} is inconsistent with \eqref{eq:itpfinite} for $v_\Omega \neq 1$. The inconsistency is due to the premise of Theorem \ref{th:poi}, \ie, finite time. A similar remark is given in \cite[Remark 9.8.1]{refWhitt}.\label{foot:inc}}
 \begin{proof}
First, for $K=\infty$, the ccdf of inter-arrival times of $A^n(\bar\Omega t)$ now satisfies $ F_{\Omega}^c (\bar\Omega x) \sim x^{-\alpha} \bar\Omega^{-\alpha}  \slow\left( \bar\Omega x \right) $ due to its scaling. From $\expectation [A^n(\bar\Omega t)] = t$, the mean inter-arrival time is {\bf one}. It follows from the underlined assumption that $\bar\Omega^{-\alpha}  \slow\left( \bar\Omega x \right)$ does not scale with $N$ and it is a slowly-varying function of $x$. Applying \cite[Theorem 1]{refLimitFractal} yields that $\widetilde{A}(t)/\zeta$ weakly converges to the process in \eqref{eq:itpinfinity}.

For the rest of cases, (i) $K<\infty$ and (ii) $K=\infty$ and $\alpha\in(2,\infty)$, we do not need any assumption because $\expectation[ \Omega^2]<\infty$ holds both for (i) and (ii) by appealing to Theorem \ref{th:heavy}. These finite variance cases were analyzed in \cite[Section 2.1.3(a)]{refFBMTaqqu} whose `ON/OFF source model' reduces to our model if we use $\mu_1 =1$ and $\mu_2 \gg 1$. Remark that it is discussed in \cite[Section 2.3]{refFBMTaqqu} that the order of limit operations does not matter in these cases.
 \end{proof}

The phrase `as $\zeta \to \infty$ and $N \to \infty$' is pregnant with meaning. The fundamental strength of the above theorem for the case of \eqref{eq:itpinfinity} is in that its result is not subject to the order of limit operations. Instead, the {\it scaling structure} between $\zeta $ and $N$, represented by $c$, determines the kind of the approximation in the sense that, as $c\to0$ and $c\to\infty$, $c^{1/\alpha} \mathbf{Y}_{\alpha} ( \frac{ t}{c} ) $ and $c^{H} \mathbf{Y}_{\alpha} ( \frac{ t}{c} ) $ respectively converges to $\mathbf{\Lambda}_\alpha(t)$ ($\alpha$-stable L\'evy motion) and $\mathbf{B}_H(t)$ (fractional Brownian motion of index $H=(3-\alpha)/2$), up to constants \cite{refPoissonBridge}. For finite $c \in(0,\infty)$, $\mathbf{Y}_{\alpha} ( \frac{ t}{c} ) $ becomes an in-between process. For the case of \eqref{eq:itpfinite}, even this scaling structure does not matter.

It is significant that $c\to0$ and $c\to\infty$ respectively equivalent to $\lim_{N\to\infty}\lim_{\zeta\to\infty}$ and $\lim_{\zeta\to\infty}\lim_{N\to\infty}$ in the literature. Therefore, the essence of the advance \cite[Theorem 1]{refLimitFractal} is that it has {\bf emancipated} the limit form of the superposition process {\it from the order of the two limit operations}, widening the applicability of the theory.
\begin{remark}\label{rem:itp}
Though, for $K=\infty$, the underlined phrase makes a strong assumption which is not reasonable in view of $\alpha = -(\log \gamma)/\log m$ which heavily depends on $N$, the above theorem deserves its result in the sense that it suggests a possible approximation of the backoff process in 802.11, based on the state-of-the-art theory. We will come back to the preciseness of the approximation later in Remark \ref{rem:precise} where we observe that $\alpha$ is required to be not too close to $1$.
\begin{asparaenum}[\rmfamily\bfseries{R\ref{rem:itp}}.1]
\item As we have discussed in Footnote \ref{foot:inc} and \cite[Chapter 9]{refWhitt} as well as at the beginning of this section, Poisson approximation in Theorem \ref{th:poi} is poor on coarse time-scales, \ie, large time. Therefore, for short-term fairness analysis, the following approximations inspired by \eqref{eq:itpinfinity} and \eqref{eq:itpfinite} are essential:
    \begin{align}
 \widetilde{A}(t)   & \approx  - {\zeta} \cdot c \cdot \mathbf{Y}_{\alpha} \left( \frac{ t}{c} \right), ~& \mbox{for}~K=\infty, ~\alpha\in(1,2), \label{eq:aitpinfinity} \\
  \widetilde{A}(t)  &  \approx  \sqrt{N \zeta} \cdot v_\Omega \cdot \mathbf{B}( t), ~  & \mbox{otherwise}. \label{eq:aitpfinite}
 \end{align}
\item {} [{Answer to \bfseries\textcolor{webmag}{Q5}}] It turns out that for $K=\infty$ and $\alpha \in (1,2)$, the superposition process $A(\zeta\bar\Omega t)=\sum_{n=1}^{N} A^n( \zeta\bar\Omega t )$ exhibits {\bf long-range dependence} due to the {\it heavy} power tail of inter-arrival times $\Omega$. This process is non-Gaussian and non-stable and has stationary, but {\it strongly dependent}, increments in the sense that it has the same covariance as a multiple of fractional Brownian motion of index $H=(3-\alpha)/2$ \cite{refPoissonBridge}. It is also shown in \cite{refPoissonBridge} that this process is (both locally and globally) {\it asymptotically self-similar} though not self-similar. We believe that networking community has been longing for a mathematical evidence which makes extensive simulations in \cite{refTickooSS} less necessary.
\end{asparaenum}
\end{remark}


Turning back to the discussion of inter-transmission probability $\probability_N [z \vert \zeta]$ in Section \ref{sec:interprobintro}, we demonstrate the strength of the above approximations in the following corollaries where $\zeta$ is now taken to be number of packets transmitted by the tagged node.

\begin{corollary}[Asymp. Inter-Transmission Probability] \label{th:aintertxprob}  \nextline
Suppose $\zeta \gg 1$ and $N \gg 1$. If $K=\infty$ along with $\alpha \in(1,2)$, we have
   \begin{align}
\probability_N [Z = z \vert \zeta] \approx  \displaystyle\int_{-\infty}^{\infty} \int_{q_-(\tau(y))}^{q_+(\tau(y))} \mathbf{Tc}^{\tau(y)/c} (x) \ud x   \cdot {\mathbf{Lv}} (y )  \ud y   \label{eq:aintertxfinaln}
\end{align}
 where $ q_{\pm} (\tau(y)) \textstyle\defeq - \left\{ z \mp \delta-(N-1)\zeta \cdot \tau(y) \right\} /(\zeta c)$, $\delta=1/2$, $
\tau(y) \defeq 1 + \zeta^{(1-\alpha)/\alpha} \slow_0 (\zeta) \cdot y $. Here $\slow_0(\cdot)$ is slowly varying at infinity, $\mathbf{Tc}^\tau (\cdot) $ is the pdf of $\mathbf{Y}_{\alpha}(\tau)$ whose cgf is given by \eqref{eq:cgf}, and $\mathbf{Lv} (\cdot )$ is the pdf of $\mathbf{S}_\alpha(1,1,0)$ whose {\bfseries index} is $\alpha =  - (\log \gamma)/{\log m}$.
\end{corollary}
\begin{proof}
Under the assumption $\zeta \gg 1$ and $N \gg 1$, it follows from Theorem \ref{th:itp} that $\sum_{n=1}^{N-1}A^n(\zeta \bar\Omega t)$ can be approximated by an Intermediate Telecom process so that its marginal distribution takes the form
\begin{align}
& \textstyle\probability \left[ \sum_{n= 1}^{N-1} A^n\left( \zeta \bar\Omega t \right) = z \right] \nonumber \\
& \approx~ \probability \left[ (N-1)\zeta t - \zeta c \mathbf{Y}_{\alpha}(t/c) \in (z-\delta,z+\delta) \right] \nonumber  \\
& = \probability \left[ \mathbf{Y}_{\alpha}(t/c) \in \left( q_-(t) ,q_+(t)  \right) \right] \nonumber \\
& = \probability \left[ \textstyle \int_{q_-(t)}^{q_+(t)} \mathbf{Tc}^{t/c} (x) \ud x    \right].  \label{eq:alastlinet}
\end{align}
 In the meantime, it follows from the definition of skewness $\beta$ and $f_\Omega (-x)=0 $, $\forall x >0$ that $$ \textstyle \beta \defeq \lim_{x\to\infty} \frac{2 F_\Omega^c (x) }{F_\Omega^c (x) + \int_{-\infty}^{-x} f_\Omega (x) dx}  - 1 = 1 .$$ Put $t=\sum_{j=1}^{\zeta} \Omega_j /(\zeta\bar\Omega) $. Applying the lesser-known stable-law central limit theorem \cite[Theorem 4.5.1]{refWhitt} to the power tailedness result of Theorem \ref{th:powertail}, taken together with the fact $\beta =1$, it follows that, for $\zeta \gg 1$,
\begin{align}
t \distapprox   1 + \zeta^{(1-\alpha)/\alpha} \cdot \slow_0 (\zeta) \cdot \mathbf{S}_\alpha(1,1,0) . \nonumber
\end{align}
 Plugging this line into \eqref{eq:alastlinet} yields \eqref{eq:aintertxfinaln}.
\end{proof}
\begin{corollary}[Inter-Transmission Probability] \label{th:intertxprob}  \nextline
Suppose $\zeta \gg 1$ and $N \gg 1$. If $K<\infty$, or $K=\infty$ along with $\alpha \in(2,\infty)$, we have
  \begin{align}
\probability_N [ z \vert \zeta] \approx  \displaystyle{\mathbf{Nm}}\left( \frac{z - (N-1)\zeta}{(N-1)\sqrt{\zeta} v_\Omega} \right)   \label{eq:intertxfinaln}
\end{align}
 where the CV $v_{\Omega}$ is given by \eqref{eq:totalcovcor}, and $ {\mathbf{Nm}} (x ) \defeq \frac{1}{\sqrt{2 \pi}} {\rm e}^{- \frac{x^2}{2}}$.
\end{corollary}
\begin{proof}
 Likewise, we have
\begin{align}
& \textstyle\probability \left[ \sum_{n= 1}^{N-1} A^n\left( \zeta \bar\Omega t \right) = z \right] \nonumber \\
& \approx~ \probability \left[ (N-1)\zeta t +  \mathbf{N} \left( 0, v_\Omega^2(N-1) \zeta t \right) \in (z-\delta,z+\delta) \right]. \nonumber  \\
& \approx~ \probability \left[ {\mathbf{Nm}} \left( \frac{z-(N-1)\zeta t}{v_\Omega \sqrt{(N-1)\zeta t}}  \right) \right] \label{eq:lastlinet}
\end{align}
where $\mathbf{N} \left(\mu,\sigma^2 \right)$ is the Gaussian random variable with mean $\mu$ and variance $\sigma^2$. Putting $t=\sum_{j=1}^{\zeta} \Omega_j /(\zeta\bar\Omega) $, $t$ is approximated by $$t \distapprox \frac{1}{\zeta} \cdot \mathbf{N} \left(\zeta,v_\Omega^2 \zeta \right) \disteq 1+ \frac{v_\Omega}{ \sqrt{\zeta}} \cdot \mathbf{N} \left(0,1\right)$$ for $\zeta \gg 1$. Thus \eqref{eq:lastlinet} becomes
\begin{align}
 \int_{-\infty}^{\infty} {\mathbf{Nm}} \left( \frac{z-(N-1) \left(\zeta+ \sqrt{\zeta}v_\Omega x \right) }{v_\Omega \sqrt{(N-1)  \left(\zeta+ \sqrt{\zeta}v_\Omega x \right) }}  \right) {\mathbf{Nm}}(x) \ud x \nonumber
\end{align}
which is approximated as \eqref{eq:intertxfinaln} because the denominator $v_\Omega (N-1)^{1/2}  (\zeta+ \sqrt{\zeta}v_\Omega x )^{1/2} $ is very large so that the first pdf of the integrand is concentrated around $z=(N-1) \left(\zeta+ \sqrt{\zeta}v_\Omega x \right)$.
\end{proof}
\begin{remark}\label{rem:intertransmission}
The derived equations provide us several penetrating insights and answers to {\bfseries\textcolor{webmag}{Q2}} as well. Note that the mean and variance of \eqref{eq:aintertxfinaln} are given by
\begin{align}
\textstyle\bar Z & \textstyle\defeq \sum_{z=0}^\infty z \cdot \probability_N  [ z \vert \zeta]  \approx (N-1) \zeta \nonumber \\
\textstyle\sigma_Z^2 & \textstyle\defeq \left( \sum_{z=0}^\infty z^2 \cdot \probability_N  [ z \vert \zeta] \right) - \bar Z^2 \approx \infty \label{eq:levyvar},
\end{align}
while those of \eqref{eq:intertxfinaln} are given by
\begin{align}
\textstyle\bar Z \approx (N-1) \zeta,  \textstyle \quad \sigma_Z^2  =(N-1)^2 \zeta \cdot v_{\Omega}^2 \label{eq:coxvar}.
\end{align}
\begin{asparaenum}[\rmfamily\bfseries{R\ref{rem:intertransmission}}.1]
\item For the case of \eqref{eq:intertxfinaln}, we can say that $Z$ is approximately Gaussian for large $\zeta$ and $N$:
 \begin{equation}\label{eq:gaussianapp}
 \textstyle Z \distapprox \mathbf{N} \left((N-1)\zeta,(N-1)^2 \zeta v_{\Omega}^2  \right)
 \end{equation}
whereupon the CV of $Z$ can be computed from \eqref{eq:coxvar} as
    \begin{equation}
v_Z \defeq  \sigma_Z / \bar Z \approx {v_{\Omega}}/{\sqrt \zeta} \label{eq:coxcovsimple}.
\end{equation}
Remarkably, we have derived the most general expression of the inter-transmission probability $\probability_N [ z \vert \zeta]$ while \cite{refBergerUniform,refMarkusExp} derived the expressions of $\probability_N [ z \vert \zeta]$ only for $N=2$.
\item \eqref{eq:aintertxfinaln} cannot be simplified in general. However, for for very large $\zeta$, hence very small $c$, it can be easily seen that $Z$ has a L\'evy $\alpha$-stable distribution. 
    Applying \cite[Proposition 2]{refPoissonBridge} to the right-hand side of \eqref{eq:itpinfinity} yields that it is negligible, implying that the inner integral of \eqref{eq:aintertxfinaln} can be removed. Then $Z$ becomes approximately {\it L\'evian} and is expressed in the form $$ \textstyle Z \distapprox \mathbf{S}_\alpha((N-1)\zeta^{\frac{1}{\alpha}}\slow_0 (\zeta),1,(N-1)\zeta). $$ This {\it manifests} the heavy-tail of $Z$, \ie, 
     \begin{equation}\label{eq:levianright}
     \textstyle \probability [ Z > x] \approx \zeta \{(N-1)\slow_0(\zeta) \}^\alpha   C_\alpha \cdot x^{-\alpha}
     \end{equation}
      where $\slow_0 (\cdot)$ is the same function used in $\tau(y)$ in \eqref{eq:aintertxfinaln} and $$ \textstyle C_\alpha = (\alpha-1)/\left({\Gamma(2-\alpha) \sin(\pi (\alpha-1)/2)}\right). $$

\item For the case of $K=\infty$ and $\alpha \in (1,2)$, the representation \eqref{eq:levianright} reveals the striking similarity between the ccdfs of $\Omega$ and $Z$. In terms of regular variation theory, both are {\it regularly varying of index $-\alpha$}, and in Mandelbrot's terminology, {\it Noah effect} of $\Omega$ {\it infiltrates} into $Z$.

\item For the case of $K=\infty$, the inter-transmission probability bifurcates into two different categories at $\alpha =2$ (or $\gamma = 1/m^2$). Plainly speaking, if $\gamma < 1/m^2$, $Z$ can still be approximated by the Gaussian distribution in \eqref{eq:gaussianapp}, otherwise 802.11 suffers from extreme unfairness containing precursors of power-tailed characteristics such as infinite variances and the {\it skewness} ($\beta=1$).

\item The skewness induces {\it leaning tendency} and {\it directional unfairness}. The leaning tendency implies the distribution is heavily leaning to the left, and the tendency increases as $\alpha$ decreases. The directional unfairness\footnotemark implies that while the right part of the inter-transmission probability $z\in (\bar Z,\infty)$ has a heavy power tail given by \eqref{eq:levianright}, its left part $z \in (-\infty,\bar Z)$ decays {\it faster} than exponentially \cite[pp.113]{refWhitt}.

\end{asparaenum}
\footnotetext{The L\'evy $\alpha$-stable law used in this work has support on the entire real line because $\alpha \in (1,2)$.}
\end{remark}

\begin{figure}[t!]
\centering
\centerline{\includegraphics[width=9.5cm]{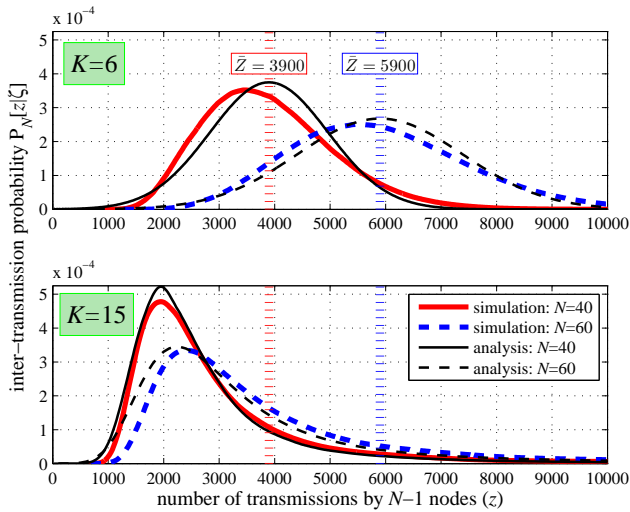}}
\caption{Inter-transmission probability $\probability_N [ z \vert \zeta]$ for $\zeta=100$; $K=6,15$; and $N=40,60$.} \label{fig:intertx}
\end{figure}

We conjecture based on extensive simulations that $\slow(\cdot)$ in \eqref{eq:ccdfregular} is approximately a constant, implying that $\slow_0(\cdot)$ in $\tau(y)$ in \eqref{eq:aintertxfinaln} is also a constant. Then it follows that the constant $\slow$ corresponds to the $y$-intercept of the straight line obtained by taking logarithms of \eqref{eq:ccdfregular}, and can be estimated from Fig. \ref{fig:power}. After manipulation akin to \cite[Theorem 4.5.2]{refWhitt}, we can show a simple relation between them: $$\textstyle \slow_0 =  (\slow/C_\alpha)^{1/\alpha}  / \bar\Omega  $$ which implies that we need to estimate only $\slow$ and $\alpha$.

\ifthenelse{\boolean{longver}}{ It is natural calculating \eqref{eq:aintertxfinaln} should take long time especially since there is no numerical methods to efficiently compute the newly discovered pdf $\mathbf{Tc}^{\tau(y)/c} (\cdot) $ which varies with $\tau(y)/c$ covering pdfs between two extremes \cite{refPoissonBridge}, \ie, L\'evy $\alpha$-stable law (for $\tau(y)/c \to \infty$) and Gaussian law (for $\tau(y)/c \to 0^+ $). What is more, \eqref{eq:aintertxfinaln} is even more complicated because $\mathbf{Tc}^{\tau(y)/c} (\cdot) $ should be integrated over a large number of intervals, \ie, $[q_-(y), q_+(y)]$, $\forall y$. We can get a handle on it only by using \cite[Theorem 4]{refNeilCDF} that lends itself to computing cdfs from oddly shaped cgfs like \eqref{eq:cgf} and an integration method called adaptive Gaussian quadrature method.}{}

In Fig. \ref{fig:intertx}, the inter-transmission probability obtained through ns-2 simulations is compared with the derived formulae of Corollaries \ref{th:aintertxprob} and \ref{th:intertxprob} for $\zeta=100$. It is significant that, for $K=6$, $\probability_N [Z = z \vert \zeta]$ is well approximated by Gaussian formula \eqref{eq:intertxfinaln} along with \eqref{eq:totalcovcor}, \eqref{eq:fixedpoint} and \eqref{eq:fixedpointg} for large $N$. This forms a striking contrast with the case $K=15$ where the distribution \eqref{eq:aintertxfinaln} is leaning to the left and its peak is {\it far apart} from its mean, \ie, $\bar Z=(N-1)\zeta$, meaning that there are {\it even heavier} tails on the right part. Our extensive simulations attested to the inevitability of complicated form \eqref{eq:aintertxfinaln}.

\begin{remark}\label{rem:precise}
{\bf Preciseness of the approximation \eqref{eq:aitpinfinity}}: remains a question due to the underlined assumption of Theorem \ref{th:itp}. Note that $N$ is determined by $\alpha = - (\log \gamma) / \log m$ provided that $\alpha$ is fixed, {\bf whereas} \cite[Theorem 1]{refLimitFractal} demands that $N\to\infty$ provided that $\alpha$ is fixed. Through extensive simulations, we have found out that the approximation \eqref{eq:aitpinfinity} becomes poor as $\alpha \to 1$ (or as $N\to\infty$). Under the above simulation setting, if $N> 80$, the approximation appears not reasonable. A thorough theory addressing this dependence between $N$ and $\alpha$ is left for future work.
\end{remark}



\section{Wavelet Analysis of Long-Range Dependence}\label{sec:longrange}

We provide simulation results to support the argument over the long-range dependence in Section \ref{sec:itpapprox} under the assumption $K=\infty$. Recall from Theorem \ref{th:itp} that the time-scaled version of the superposition arrival process is approximately
$$\textstyle A(\zeta \bar\Omega t) = \sum_{n=1}^{N} A^n(\zeta \bar\Omega t) \approx N \zeta t - \zeta c \cdot \mathbf{Y}_{\alpha} \left( \frac{ t}{c} \right) $$ which holds for $N$ such that $\alpha = - (\log \gamma)/\log m < 2$.
Note that such $N$ is to ensure $\Omega$ is {\it strict-sense heavy-tailed} (See Theorem \ref{th:powertail}). Then by appealing to \cite{refPoissonBridge}, 
one can show that $A(\zeta \bar\Omega t)$ has long-range dependent increments in the sense that
\begin{compactitem}
\item $\textstyle A(\zeta \bar\Omega t)$ has the same covariance as a multiple of fractional Brownian motion of index $H \defeq (3-\alpha)/2$.
\end{compactitem}
It is easy to see that $1/2 < H < 1$ due to $1<\alpha<2$.

All simulations obtained from {\it ns-2} simulator use a $437h$ warm-up period, after which we collected $728h$-long traces. To analyze these traces, we use the latest addition to the toolkit of inference techniques for long-range dependence, \ie, the {\it refined} wavelet-based method using {\it Daubechies} wavelets with $M$ vanishing moments which was proposed by Abry \etal \cite{refWaveletLens} They proposed the first unbiased estimator $y_j$ taking the form $$\expectation [y_j] = \log_2 \left( \expectation \left[ d_j^2  \right] \right) ,$$ considering the complication presented by the property $\expectation [\log(\cdot)] \neq \log (\expectation[\cdot])$ where $d_j$ is called {\it detail processes} of the wavelet transform. \ifthenelse{\boolean{longver}}{It is also shown that this method is applicable even to non-Gaussian processes. This offers a clear advantage to our case where $A'(t)$ is non-Gaussian.}{}

\begin{figure*}[t!]
\centering
\begin{center}
\mbox{
  \subfigure[$K=6$, $N=40$]{\label{fig:lrd6}
        \includegraphics[width=5.5cm]{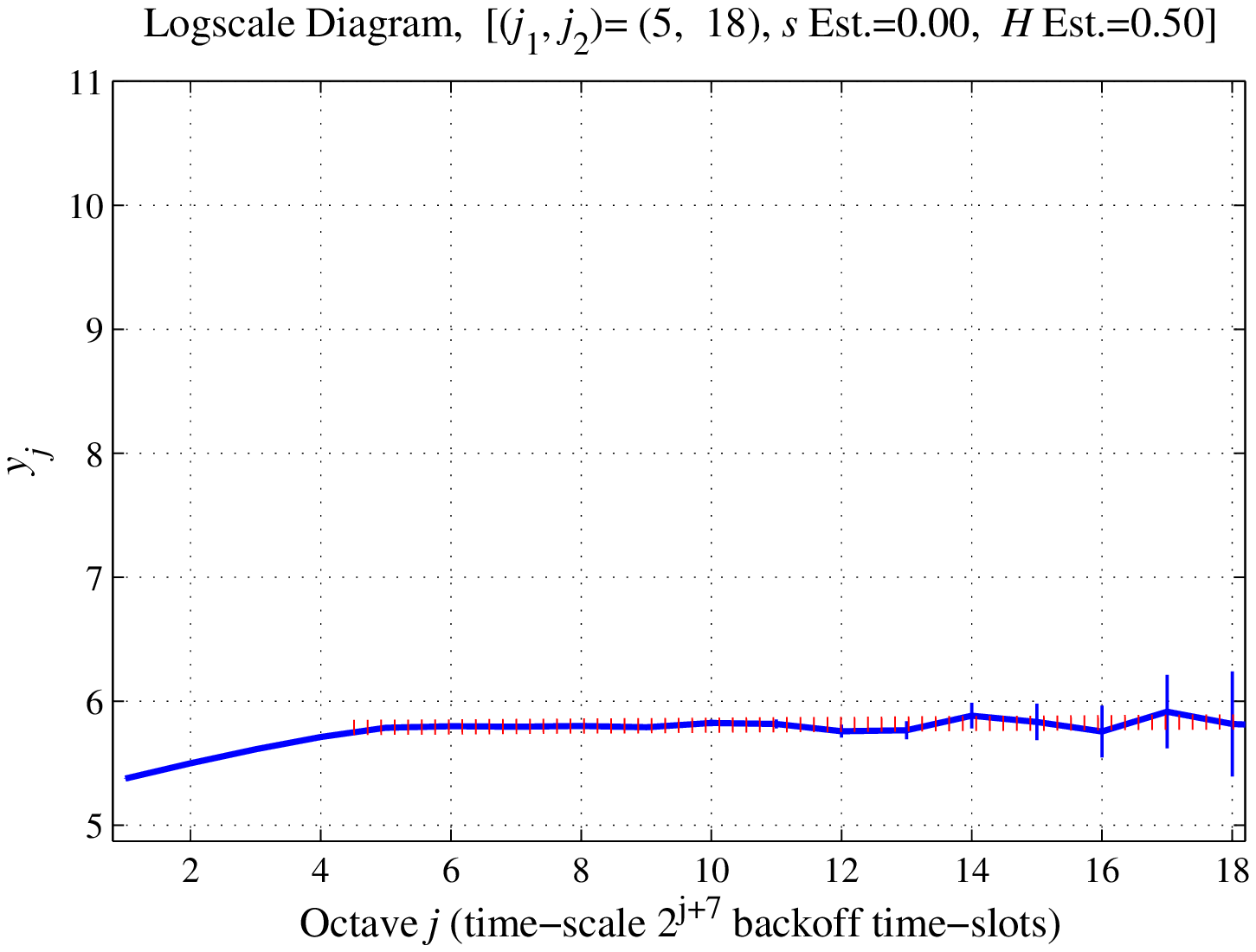}}
  \subfigure[$K=15$, $N=40$]{\label{fig:lrd15}
        \includegraphics[width=5.5cm]{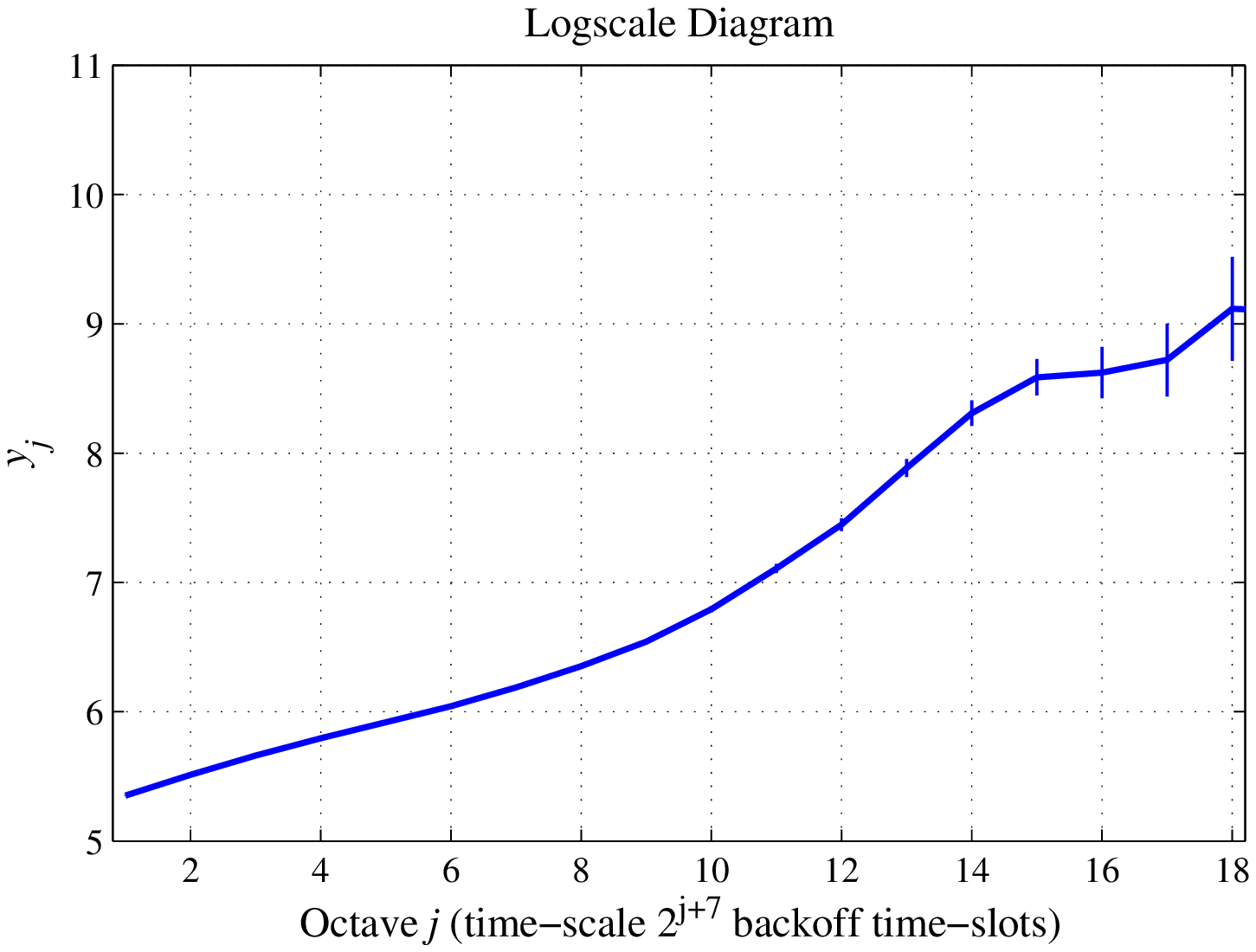}}
  \subfigure[$K=25$, $N=40$]{\label{fig:lrd25}
        \includegraphics[width=5.5cm]{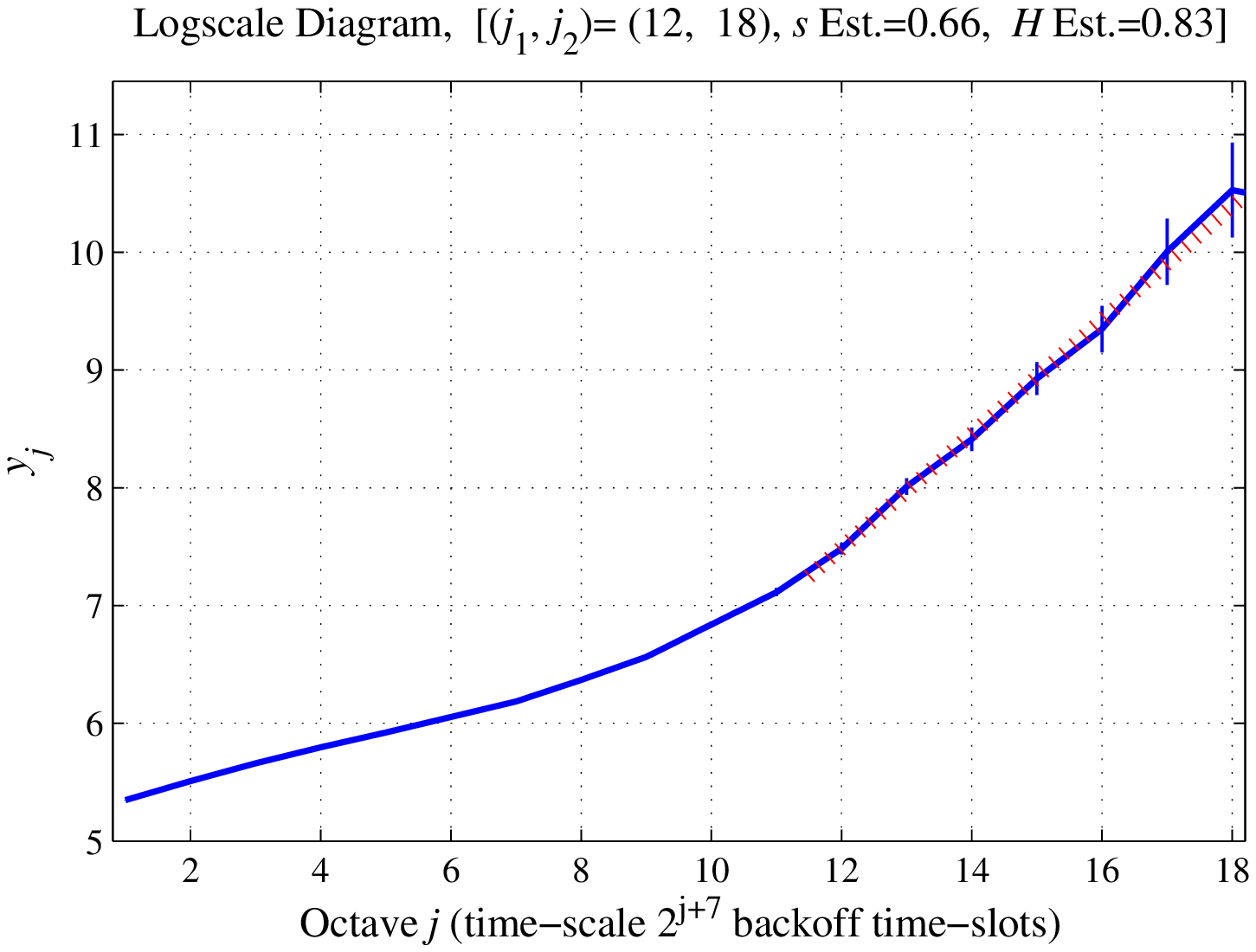}}
        }
\caption{Wavelet spectra using Daubechies wavelets with $M=2$.} \label{fig:lrd}
\end{center}
\end{figure*}

The estimates $y_j$ of the wavelet spectra over all time-scales $j$, called {\it octaves}, are shown in Fig. \ref{fig:lrd} for $K=6$, $K=15$ and $K=25$. Here we fix the other parameters as $N=40$ and $M=2$. Though we present here only the simulation results using Daubechies wavelets with $M=2$, we obtained similar results using Daubechies wavelets with $M>2$ and {\it Discrete Meyer} wavelets. To quantify the integrity of the method, Gaussian $95\%$ confidence intervals corresponding to the variability of $y_j$ are also shown as the vertical segments centered on the estimates $y_j$. Then the measurement of index $H$, called {\it Hurst} parameter, is reduced to the identification of region of {\it alignment}, the determination of the its lower and upper cutoff octaves, $j_1$ and $j_2$, respectively, and the determination of the slope over the alignment region which we denote by $\bar s$. From the slope estimate $\bar s$, we can obtain the estimates of $H$ from the formula $$\bar H \defeq (1+ \bar s )/2 . $$
\ifthenelse{\boolean{longver}}{Decisions whether the alignment region is aligned or not were made based on the Chi-squared goodness of fit test \cite{refWaveletScaling}. Note that the extent of long-range dependence increases with $\bar H$, or equivalently $\bar s$. We also let $s_j$ denote the slope at octave $j$.}


Fig. \ref{fig:lrd25} demonstrates that, for the case $K=25$, the superposition arrival process possesses a sustained correlation structure over a broad range of time-scales $j \in [1, 18]$ where $s_j$ converges to $0.66$ at octave $j=18$, whereas, for the case $K=6$, it shows a weaker correlation structure over a narrow range $j \in [1, 5]$ as shown in Fig. \ref{fig:lrd6}. The estimate of $H$ for $K=25$ over the alignment region $(j_1, j_2) = (12, 18)$ approaches $\bar H = 0.83$ around $(16,17)$ which approximately matches with analytical formula $H=(3-\alpha)/2 = 0.90$ where $\alpha$ is obtained from $\alpha = (\log \gamma)/\log m$, Eqs. \eqref{eq:fixedpoint} and \eqref{eq:fixedpointg}. The slope estimate over the alignment region for $K=6$ is computed as $\bar H = 0.50$, implying that long-range dependence is not observed. A striking observation that can be made by comparing Figs. \ref{fig:lrd6} and \ref{fig:lrd15} with Fig. \ref{fig:lrd25} is that the per-octave slope $s_j$ increases as octave $j$ increases and convergent only if $K$ is large enough as in Fig. \ref{fig:lrd25}. \ifthenelse{\boolean{longver}}{Even for large $K$ as shown in Fig. \ref{fig:lrd25}, the slope is small for low octaves.}{}

\begin{obs}[LRD over coarse times scales]\label{obs:time}  \nextline
 Long-range dependence of the superposition process is conspicuous only over coarse time-scales. 
\end{obs}

\begin{remark}\label{rem:lrd}
\ifthenelse{\boolean{longver}}{ Here the word `conspicuous' is used in the sense that the per-octave estimate $H_j \defeq \left( 1+s_j \right) /2$ closely matches with the theoretical value $H$ computed from the analytical formulae.}{} Essentially, there are two reasons behind this phenomenon which also give us answers to {\bfseries\textcolor{webmag}{Q5}}.

\begin{asparaenum}[\rmfamily\bfseries{R\ref{rem:lrd}}.1]
\item {\bf Per-node process slows down}: It is important to recall that, for $K=\infty$, we first established Poisson process approximation for the superposition process in Theorem \ref{th:poi}, meaning that we {\bf cannot} observe long-range dependence on normal time-scales. As is the constant intensity of the superposition process for Theorem \ref{th:poi}, the constant intensity of the {\bf component} process is essential for Theorem \ref{th:itp}. To satisfy the latter, we had to consider $A^n(\zeta \bar\Omega t)$ instead of $A^n( t)$ because $A^n(t)$ becomes sparser as $N\to\infty$. That being said, we must view the superposition process over {\it coarse} time-scales $\zeta \bar\Omega t$ instead of $t$ to satisfy the premise of Theorem \ref{th:itp}, which explains long-range dependence.
\item  {\bf Additional scaling of time}: Another assumption of the limit regime considered in \cite{refLimitFractal} is $\zeta \to \infty$ at the same time as $N \to \infty$. This implies we need additional scaling of time to compensate for the scaling of space.
\item
    \ifthenelse{\boolean{longver}}{Due to the above two speedups which require a even coarser time-scales, we can observe long-range dependence of aggregate total load only over coarse time-scales.}{} In practical terms, if the wireless link capacity is {\it shared} by many nodes, the aggregate transmission process is highly
     invulnerable to long-range dependence for most practical $K$ values, essentially due to {\it reduced per-node rate} and {\it additional time scaling}.
\end{asparaenum}
\end{remark}

\ifthenelse{\boolean{longver}}{ It is important to note that the arrival process of each {\it individual} flow possesses long-range dependence with $H = (3-\alpha)/2$ if the inter-arrival times of each individual process is heavy-tailed \cite{refLipsky}. On the contrary, long-range dependence of the superposition arrival process $A(\zeta \bar\Omega t)$ is much weaker than that predicted by theory in that $ H_j < H $, or equivalently $ s_j < 2 - \alpha$ for low octaves $j$, and regains its influence only for high octaves $j$.}{}

We also conjecture that the above coarser time scalings caused the empirical analyses of Veres and Boda \cite{refChaoticVeres} (in the context of TCP) and Tickoo and Sikdar \cite{refTickooSS} (in the context of 802.11) not to support long-range dependence of the superposition arrival process of TCP sources --- they observed that $\hat H \approx 0.5$ (or $\hat s = 0$), implying short-range dependence. This is because both 802.11 nodes accessing a common base station and TCP flows traversing a common bottleneck link (i) have similar backoff mechanisms and (ii) reduce (or slow down) their transmission rates to share the given capacity as the population increases. \ifthenelse{\boolean{longver}}{It is interesting that this simple analogy constitutes the fundamental causes of the absence of long-range dependence identified in Observation \ref{obs:time}.}{} \ifthenelse{\boolean{longver}}{A riddle is solved.}

\section{Concluding Remarks}\label{sec:conclusion}
Beginning with derivation of per-packet backoff distribution, based on which we studied its coefficient of variation that plays a key role in formulating short-term fairness in later sections, we have conducted a rigorous analysis of the backoff process in 802.11 and provided answers to several open questions. 

The power-tail principle states that the per-packet backoff has a truncated {\bf Pareto-type} tail distribution, a simplistic description elucidating existing works. This in turn indicates that its heavy-tailedness in the strict-sense inherits from {\it collision} and paves the way for the rest of analysis. The {\bf dichotomy} of aggregation, proven with the aids of a recent advance \cite[Theorem 1]{refLimitFractal} in probabilistic community, now tells the whole story of contrary limits of the superposition process, \ie, Poisson process and Intermediate Telecom process, emphasizing the importance of time-scales on which we view the backoff processes. Thanks to the applicability of \cite{refLimitFractal} widened by the {\bf order-free} scaling operations of time ($\zeta$) and population ($N$), we identified {\it long-range dependence} in 802.11 and discovered that the inter-transmission probability bifurcates into two categories: either approximately Gaussian or a complicated distribution which, under a limiting condition, simplifies to {\it L\'evy} $\alpha$-stable distribution with $\alpha \in (1,2)$ possessing strong power-tail characteristics.

Though we have also conducted empirical analysis using wavelet-based method to support long-range dependence behavior inherent in 802.11, since we are with Willinger \etal \cite{refWil} on the point --- of cardinal importance is to advance our genuine physical understanding applicable to many other systems, we believe that the essence of our analysis of long-range dependence lies in its mathematical explanation for the behavior. That is, the heavy-tailed inter-arrival time of each per-node transmission process causes long-range dependence of the aggregate transmission process at the base station though this dependence is seldom observed.

These results explore the fundamental principles characterizing the backoff process in 802.11. Some of them recall to our mind the beauty of simplicity, governing the asymptotic dynamics of 802.11, and the others form the theoretical groundwork of short-term fairness. 


\section*{Acknowledgment}
The authors would like to thank Jean-Yves Le Boudec for helpful discussions and penetrating comments which have improved the quality of this work. We also would like to thank Ingemar Kaj for elaborating upon his work. Lastly, we would like to thank Hans Alm{\aa}sbakk for making it possible to run very long simulations on a cluster of a dozen processors.

\bibliographystyle{abbrv}
\bibliography{IEEEabrv,bib2}

\appendix
\section{Appendix: Proofs}

\subsection{Derivation of \eqref{eq:totalcovcor}}\label{sec:dertotal}
Plugging $b_k = {m^k} {b_0}$ into $v_{\Omega} = \sigma_{\Omega}/\bar \Omega$  yields
 \begin{equation}
v_{\Omega} =  \sqrt{\textstyle\frac{\Upsilon }{ \left( \sum_{k \in \K} \textstyle(b_0 m^k - 1/2 ) \gamma^k \right)^2 } -1 }. \label{eq:temptotalcov}
\end{equation}
Here the nominator inside the square root is simplified as
\begin{align}
&\Upsilon \defeq \textstyle \left( \sum_{k=0}^K (b_0 m^k - 1/2 )^2 \gamma^k (1+ v_k^2)\right) \nonumber \\ & +  2 \sum_{k=1}^K (b_0 m^k - 1/2 ) \gamma^k \left( b_0 \frac{m^k-1}{m-1}  -k/2 \right) \nonumber \\ & \textstyle = \left( \displaystyle\sum_{k=0}^K \textstyle \left( \frac{m+1}{m-1} +v_k^2 \right) (b_0 m^k - 1/2 )^2 \gamma^k \right) \nonumber \\ & \textstyle - \frac{2}{m-1}(b_0-1/2)^2 - \displaystyle\sum_{k=1}^K \textstyle(b_0 m^k - 1/2 ) \gamma^k (k+\frac{2b_0-1}{m-1})\nonumber\\
& = \textstyle\sum_{k=0}^K \textstyle \left( \frac{m+1}{m-1} +v_k^2 \right) \left(b_0 m^k - \frac{1}{2} \right)^2 \gamma^k \nonumber \\ & \textstyle - \left( k+\frac{2b_0-1}{m-1} \right) \left(b_0 m^k - \frac{1}{2} \right) \gamma^k . \nonumber
\end{align}
Plugging the last line into \eqref{eq:temptotalcov} yields \eqref{eq:totalcovcor}.

\subsection{Proof of Theorem \ref{th:powertail}}
Throughout the proof, we denote the sets of real numbers (positive real numbers), integers (positive integers), and rational numbers by $\R$ ($\R^+ $), $\Z$ ($\Z^+$) and $\Q$, to simplify the exposition. Denoting the LST of $f_i (b)$ by $F_i (s)$, we begin the proof by considering the LST of \eqref{eq:totalpdf}:
\begin{equation}\label{eq:lsttotalpdf}
\textstyle F_{\Omega} (s) =  \textstyle \frac{ 1 - \gamma}{\gamma} \underbrace{\sum_{k =0}^{\infty} \left\{ \prod_{i=0}^{k} \gamma F_i (s) \right\} }_{G(s) }.
\end{equation}
This is an infinite sum of the products of
\begin{equation}
F_i (s) = \frac{  1- \exp\left(-(2b_0 m^i - 1) s \right) }{  (2b_0 m^i - 1) s } \label{eq:lstuni}
\end{equation}
that is the LST of the uniform distribution with mean $b_0 m^i - 1/2$. For notational simplicity, we adopt the change of variable $x \defeq 2b_0 s$ such that $x$ also belongs to $\R^+$. Thus we have
\begin{equation}\label{eq:lsttotalpdf2}
 G (x) =  \sum_{k =0}^{\infty} \left\{ \prod_{i=0}^{k} \gamma \frac{  1- \exp\left(-(m^i - 1/(2b_0)) x \right) }{  ( m^i - 1/(2b_0)) x} \right\} .
\end{equation}
Since $F_i(x) < 1$ for $x \in \R^+$, it is easy to see that $G(x)$ is convergent on $\R^+$. Then it follows from {\it Bernstein's Theorem} \cite[pp.439]{refFellerProb2} that $G(x)$ is {\it completely monotone}. That is, $G(x) > 0 $ and it {\bf has} derivatives of all orders, which satisfy
\begin{equation}\label{eq:compmonotone}
\textstyle (-1)^i \frac{\ud^i G}{\ud x^i} (x) > 0,~ \forall i \in \Z^+ ,
\end{equation}
which implies that the $i$th derivative of $G (x)$ is {\it strictly monotone} for all $i \in \Z^+$.

\vspace{1mm}
{\bfseries\rmfamily Step 1: Recursive relation in $G (\cdot)$}

The crucial observation that paves the way for applying the theory of regular variation \cite{refRegularBingham} is the following {\it recursive relation} hidden in the underbraced term of \eqref{eq:lsttotalpdf}:
\begin{equation}\label{eq:lstgx}
 G(x) = \gamma F_0 (x)  \left\{ 1 + H(mx) \right\}
\end{equation}
where
\begin{equation}\label{eq:lsttotalpdf3}
 H (x) =  \sum_{k =0}^{\infty} \left\{ \prod_{i=0}^{k} \gamma \frac{  1- \exp\left(-(m^i - 1/(2b_0 m)) x \right) }{  ( m^i - 1/(2b_0 m)) x} \right\} .
\end{equation}
Let $\alpha = -({\log\gamma})/{\log m} \in \R^+$ and $z \defeq \ceiltheta \in \Z^+$ which designates the smallest integer not less than $\alpha$. It follows from $\alpha >0 $ that $z  \geq 1 $. Appealing to Theorem \ref{th:heavy} and the basic property of the LST, \ie, $\lim_{s\to 0^+} \frac{\ud^z F_{\Omega}}{\ud s^z} = (-1)^z \expectation[ \Omega^z ]$ for $z \in \Z^+ $, it follows that
\begin{align}\nonumber
\textstyle \lim_{x\to 0^+} \frac{\ud^z G}{\ud x^z} (x) & = \textstyle\frac{1}{(2b_0)^z} \lim_{s\to 0^+ } \frac{\ud^z G}{\ud s^z} (s)\\
&= \textstyle\frac{\gamma}{(2b_0)^z (1-\gamma)} (-1)^z \expectation[ \Omega^z ] = (-1)^z \cdot \infty . \nonumber
\end{align}
 Recall $\lim_{x\to 0^+} \frac{\ud^{i} G}{\ud x^{i}} (x)$ is finite for $i < z$ by Theorem \ref{th:heavy}. Likewise, it is easy to see that $\lim_{x\to 0^+} \frac{\ud^{i} H}{\ud x^{i}} (x)$ is finite for $i < z$ and infinite for $i\geq z$. Moreover, since both the nominator and denominator of the limit
\begin{align}
 \lim_{x\to 0^+} \frac{\frac{\ud^z G}{\ud x^z} (x)}{\frac{\ud^z H}{\ud x^z} (x)}  \label{eq:ratiodiff}
\end{align}
are infinite, we may remove arbitrary number of products whose $z$th derivatives at $x=0^+$ are finite from both of $G(x)$ and $H(x)$, implying in turn that we may replace the summation operation $\sum_{k=0}^{\infty}$ in \eqref{eq:lsttotalpdf2} and \eqref{eq:lsttotalpdf3} by $\sum_{k=k'}^{\infty}$ for any $k' \geq 0$. Formally speaking, we have
\begin{align}
\frac{\ud^z }{\ud x^z} G(x) & = \left\{  \frac{\ud^z }{\ud x^z} \sum_{k =0}^{k'-1} \prod_{i=0}^{k} \gamma F_i(x) \right\} \nonumber \\ & + \left\{ \frac{\ud^z }{\ud x^z} \prod_{i=0}^{k'-1} \gamma F_i(x) \right\} \left\{ \sum_{k =k'}^{\infty}  \prod_{i=k'}^{k} \gamma F_i(x) \right\} \nonumber \\
& + \left\{  \prod_{i=0}^{k'-1} \gamma F_i(x) \right\} \left\{ \frac{\ud^z }{\ud x^z} \sum_{k =k'}^{\infty}  \prod_{i=k'}^{k} \gamma F_i(x) \right\}  \label{eq:ratiodiff2}
\end{align}
where only the third term \eqref{eq:ratiodiff2} becomes infinite as $x \to 0^+$. Therefore, we can easily see that the difference between $G(x)$ and $H(x)$ vanishes as $k'$ increases and hence \eqref{eq:ratiodiff} must be $1$.

Taking derivatives of both sides of \eqref{eq:lstgx} $z$ times and after some manipulation, it becomes clear that it is sufficient to consider only infinite terms which are related to each other in the following form:
\begin{align}
\textstyle h(m) & \defeq \lim_{x\to 0^+} \frac{\frac{\ud^z G}{\ud x^z} (mx)}{\frac{\ud^z G}{\ud x^z} (x)} = \lim_{x\to 0^+} \frac{\frac{\ud^z H}{\ud x^z} (mx)}{\frac{\ud^z G}{\ud x^z} (x)} \lim_{x\to 0^+} \frac{\frac{\ud^z G}{\ud x^z} (mx)}{\frac{\ud^z H}{\ud x^z} (mx)} \nonumber \\ & \textstyle = m^{-z} \gamma^{-1} = m^{\alpha -z} \label{eq:quantifierm}
\end{align}
where we also exploited the fact that \eqref{eq:ratiodiff} is $1$.  Because the convergence of \eqref{eq:quantifierm} holds for any real sequences of $x_k \to 0^+$, we have that $h(y) = y^{\alpha-z}$ for $y \in  \M$ where $\M \defeq \{ m^i ~\vert~ i \in \Z \}$ is a countably infinite set that is {\it nowhere} dense in $\R^+$. The set on which the relation $h(y)=y^{\alpha-z}$ holds is often baptized {\it quantifier set} in regular variation theory.

\vspace{1mm}
{\bfseries\rmfamily Step 2: Quantifier set is dense in $\R^+$}

 We will show that $h(y)=y^{\alpha-z}$ holds on a dense subset $\Lset$ of $\R^+$. Define a set $$\textstyle\Lset \defeq \{ \lambda \in \R^+ ~\vert~ (\log \lambda)/\log m \in \R  \backslash \Q \}$$ where $\R  \backslash \Q$ is the set of irrational numbers. It should be clear that $\M$ and $\Lset$ are disjoint, \ie, $\M \cap \Lset = \emptyset$ and the set $\Lset$ is {\it dense} in $\R^+$ because it can be rewritten as $\Lset = \{ m^{{y}} \in \R^+ ~\vert~ y \in \R  \backslash \Q \}$. Defining
 $$ \textstyle \Upsilon(y,x) \defeq \frac{\ud^z G}{\ud x^z} (y x) / \frac{\ud^z G}{\ud x^z} (x),$$
 we can see that $\Upsilon(y,x)$ is strictly decreasing in $y$ because it follows from \eqref{eq:compmonotone}, \ie, complete monotonicity, that
 $$ \textstyle \frac{\ud \Upsilon(y,x)}{\ud y} = \frac{\ud (yx)}{\ud y} \frac{\ud^{z+1} G}{\ud x^{z+1}} (y x) / \frac{\ud^z G}{\ud x^z} (x)  < 0, ~\forall z\in \Z^+ .$$
 Pick $\lambda \in \Lset$ in the interval $(m^i, m^{i+1})$ for {\it any} $i \in \Z$. Since $\Upsilon(y,x)>0$ is strictly decreasing in $y$, it is upper-bounded by $ m^{i(\alpha-z)}$ as $x\to 0^+$, meaning that $\Upsilon(y,x)$ is ultimately bounded in $x$. From its series expansion, it is easy to see that it is ultimately monotone in $x$ as $x\to 0^+$. Then we can apply \cite[Theorem 3.14]{refAnalysis} to show that there exists $\tilde\alpha$ such that
\begin{equation}\label{eq:quantifierl}
\textstyle h(\lambda) = \lim_{x\to 0^+} \Upsilon(\lambda,x) = \lambda^{\tilde\alpha-z} ,
\end{equation}
which in turn implies that $h(\lambda^j) = \lambda^{j(\tilde\alpha-z)}$, $\forall j \in \Z$, as \eqref{eq:quantifierm} did. {\it Assume that $\tilde\alpha \neq \alpha$}. Because $\Upsilon(y,x)$ is strictly decreasing in $y$, irrespective of $z$, we have
\begin{equation}\label{eq:mcondition}
\textstyle m^{\alpha-z}  \leq \lim_{x\to 0^+} \Upsilon(y,x)  \leq 1
\end{equation}
for $y \in (1, m)$. Put $\hat y \defeq m^{-\lfloor j (\log\lambda)/\log m \rfloor} \lambda^j$ for $j \in \Z$. This can be rearranged as $$\hat y= m^{j(\log\lambda)/\log m - \lfloor j (\log\lambda)/\log m \rfloor}$$ and $(\log \lambda ) / \log m$ is irrational, hence its exponent is on $(0,1)$ and $\hat y$ is on the interval $(1,m)$. We now have from \eqref{eq:quantifierm} and \eqref{eq:quantifierl} that
\begin{align}
\textstyle \lim_{x\to 0^+} \Upsilon(\hat y,x) & = \textstyle \lim_{x\to 0^+} \frac{\frac{\ud^z G}{\ud x^z} ( m^{-\lfloor \frac{j \log\lambda}{\log m} \rfloor} \lambda^j  x)}{\frac{\ud^z G}{\ud x^z} (\lambda^j x)} \cdot  \frac{\frac{\ud^z G}{\ud x^z} ( \lambda^j x)}{\frac{\ud^z G}{\ud x^z} (x)}  \nonumber \\
 & = m^{-\lfloor \frac{j \log\lambda}{\log m} \rfloor (\alpha-z )} \lambda^{j (\tilde\alpha-z)} \nonumber \\
 & =  m^{-\lfloor \frac{j \log\lambda}{\log m} \rfloor (\alpha-z ) + \frac{j \log\lambda}{\log m} (\tilde\alpha-z)} \nonumber\\
 & = m^{\left( \frac{j \log\lambda}{\log m} - \lfloor \frac{j \log\lambda}{\log m} \rfloor \right) (\alpha-z )} \cdot \lambda^{ j(\tilde\alpha-\alpha)} \nonumber.
\end{align}
where the key point is that the second equality follows from $\M \cap \Lset = \emptyset$. Since the last term belongs to the closed interval $$\textstyle \mathcal{I}(j) \defeq [m^{\alpha - z} \lambda^{ j(\tilde\alpha-\alpha)}, \lambda^{ j(\tilde\alpha-\alpha)}]$$ and $\tilde\alpha \neq \alpha$, we must be able to pick $j \in \Z$ such that $\mathcal{I}(j)$ does not overlap with $[m^{\alpha-z},1]$. \ifthenelse{\boolean{longver}}{In other words, \eqref{eq:mcondition} does not hold any longer.}{}
This proves by contradiction that $h(\lambda)=\lambda^{\alpha-z}$ holds for $\lambda \in \Lset$ that is dense in $\R^+$.

\vspace{1mm}
{\bfseries\rmfamily Step 3: Applying regular variation theory}

Applying the `Karamata Theorem for monotone functions' \cite[Theorem 1.10.2]{refRegularBingham} to the conclusion we obtained in Step 2 establishes that $\frac{\ud^z G}{\ud x^z} (x)$ is {\it regularly varying} (on the right) at the origin $x=0$ with index $\alpha-z$. Formally speaking, $G(s)$ satisfies
\begin{equation}\label{eq:diffgsregular}
\textstyle\frac{\ud^z G}{\ud s^z} (s) \sim s^{\alpha-z} \slow^* \left( \frac{1}{s} \right) \quad \textrm{ as}~s \to 0^+,
\end{equation}
 where $\slow^*\left( x \right)$ is slowly varying at infinity $x=\infty$, \ie, $$\lim_{x \to \infty} \slow^*\left( yx \right)/\slow^*\left( x \right) = 1$$ for all $y \in \R^+$. Note that the original Karamata Tauberian Theorem in \cite[Theorem 1.7.1]{refRegularBingham} and \cite[pp.445]{refFellerProb2} cannot be applied due to the fact $\alpha - z \leq 0$. These theorems are complemented by the modified Karamata Tauberian Theorem in \cite[Theorem 8.1.6]{refRegularBingham} and \cite{refTauberian}, which we apply to \eqref{eq:diffgsregular} to show \eqref{eq:ccdfregular}.

\end{document}